\pdfoutput=1
\documentclass[11pt]{article}
\usepackage[T1]{fontenc}
\usepackage{amsmath,amssymb,amsthm}
\usepackage{newtxtext,newtxmath}
\usepackage[letterpaper,margin=1in]{geometry}
\usepackage{graphicx,array,longtable,float,placeins}
\usepackage[numbers,sort&compress]{natbib}
\usepackage{url}
\usepackage[hidelinks]{hyperref}
\usepackage{microtype}
\newtheorem{theorem}{Theorem}
\newtheorem{lemma}[theorem]{Lemma}
\newtheorem{proposition}[theorem]{Proposition}
\newtheorem{corollary}[theorem]{Corollary}
\theoremstyle{remark}
\newtheorem{remark}[theorem]{Remark}

\newcommand{\floorFracMaxPct}{90}
\newcommand{\floorFracMinPct}{21}
\newcommand{\nFamiliesOpen}{10}
\newcommand{\nModels}{28}
\newcommand{\nModelsFrontier}{8}
\newcommand{\nModelsOpen}{20}
\newcommand{\nVotesBaseExact}{174,384}
\newcommand{\rhoCrossB}{0.60}
\newcommand{\rhoCrossG}{0.58}
\newcommand{\rhoSelfMedian}{1.00}
\newcommand{\rhoSelfTypical}{0.84}
\newcommand{\robustAblationN}{73,410}
\newcommand{\robustBootB}{2000}
\newcommand{\robustExchangeRMSE}{0.148}
\newcommand{\robustExchangeRtwo}{0.906}
\newcommand{\robustFullRMSE}{0.088}
\newcommand{\robustFullRtwo}{0.967}
\newcommand{\robustIndepRMSE}{0.193}
\newcommand{\robustIndepRtwo}{0.840}
\newcommand{\selfGainLoss}{0.001}
\newcommand{\SelectionBootstrapValid}{1000}
\newcommand{\SelectionMonteCarloDraws}{65536}
\newcommand{\SelectionPrimaryFullScaledLoss}{50.77}
\newcommand{\SelectionPrimaryIndependenceScaledLoss}{52.50}
\newcommand{\SelectionPrimaryMajorityScaledLoss}{60.25}
\newcommand{\SelectionPrimaryAbsoluteBenefitVsMajorityScaledUnits}{9.48}
\newcommand{\SelectionPrimaryAbsoluteBenefitVsIndependenceScaledUnits}{1.73}
\newcommand{\SelectionPrimaryRelativeBenefitPct}{15.73}
\newcommand{\SelectionPrimaryAbsoluteBenefitVsMajorityCILowScaledUnits}{7.87}
\newcommand{\SelectionPrimaryAbsoluteBenefitVsMajorityCIHighScaledUnits}{10.26}
\newcommand{\SelectionPrimaryAbsoluteBenefitVsIndependenceCILowScaledUnits}{0.68}
\newcommand{\SelectionPrimaryAbsoluteBenefitVsIndependenceCIHighScaledUnits}{2.33}
\newcommand{\SelectionPrimaryRelativeBenefitCILowPct}{13.41}
\newcommand{\SelectionPrimaryRelativeBenefitCIHighPct}{16.75}
\newcommand{\SelectionIndependenceBenefitVsMajorityScaledUnits}{7.75}
\newcommand{\SelectionIndependenceBenefitVsMajorityCILowScaledUnits}{6.34}
\newcommand{\SelectionIndependenceBenefitVsMajorityCIHighScaledUnits}{8.86}
\newcommand{\SelectionThresholdOnlySharePct}{81.73}
\newcommand{\SelectionDependenceSharePct}{18.27}
\newcommand{\SelectionFullIndependenceThresholdDisagreementPct}{40.90}
\newcommand{\SelectionFullIndependenceThresholdDisagreementCILowPct}{36.19}
\newcommand{\SelectionFullIndependenceThresholdDisagreementCIHighPct}{43.99}
\newcommand{\SelectionListwisePointRelativeBenefitPct}{15.73}
\newcommand{\SelectionAgreementOnlyPointRelativeBenefitPct}{15.34}
\newcommand{\calibRsqCellMax}{0.98}
\newcommand{\calibRsqCellMedian}{0.94}
\newcommand{\calibRsqCellMin}{0.66}
\newcommand{\calibRsqDomMax}{0.98}
\newcommand{\calibRsqDomMin}{0.86}
\newcommand{\didFactsDelta}{+0.13}
\newcommand{\didFactsDeltaCI}{[-0.03, 0.34]}
\newcommand{\didTruthDelta}{+0.09}
\newcommand{\didTruthDeltaCI}{[-0.05, 0.37]}
\newcommand{\factorOffdiagMax}{0.64}
\newcommand{\factorOffdiagMin}{0.40}
\newcommand{\factorOffdiagTyp}{0.55}
\newcommand{\factorResidMax}{0.13}
\newcommand{\factorResidMin}{0.08}
\newcommand{\factorResidTyp}{0.11}
\newcommand{\lossFloorMaxPct}{94}
\newcommand{\lossFloorMinPct}{68}
\newcommand{\mdeFrFacts}{0.21}
\newcommand{\mdeFrTruth}{0.25}
\newcommand{\nCrossPairs}{363}
\newcommand{\nVotesReasonExact}{29,784}
\newcommand{\rhoCrossMax}{0.85}
\newcommand{\rhoCrossMin}{0.24}
\newcommand{\robustExchangeBias}{-0.014}
\newcommand{\robustExchangeRtwoCI}{[0.902, 0.910]}
\newcommand{\robustFullBias}{-0.014}
\newcommand{\robustFullRtwoCI}{[0.965, 0.968]}
\newcommand{\robustIndepBias}{-0.014}
\newcommand{\robustIndepRtwoCI}{[0.832, 0.848]}
\newcommand{\robustLodoFullRange}{0.954--0.977}
\newcommand{\robustLodoGainRange}{0.086--0.132}
\newcommand{\robustLodoIndepRange}{0.757--0.869}
\newcommand{\robustLofoFullRange}{0.952--0.969}
\newcommand{\robustLofoGainRange}{0.090--0.112}
\newcommand{\robustLofoIndepRange}{0.803--0.855}
\newcommand{\robustLomoFullRange}{0.953--0.969}
\newcommand{\robustLomoGainRange}{0.104--0.113}
\newcommand{\robustLomoIndepRange}{0.765--0.852}
\newcommand{\rsnFrBalaccMax}{0.93}
\newcommand{\rsnFrBalaccMin}{0.55}
\newcommand{\rsnFrFactsB}{0.84}
\newcommand{\rsnFrFactsBaseB}{0.90}
\newcommand{\rsnFrFactsDelta}{-0.07}
\newcommand{\rsnFrFactsDeltaCI}{[-0.21, 0.09]}
\newcommand{\rsnFrTruthB}{0.76}
\newcommand{\rsnFrTruthBaseB}{0.78}
\newcommand{\rsnFrTruthDelta}{-0.02}
\newcommand{\rsnFrTruthDeltaCI}{[-0.15, 0.21]}
\newcommand{\rsnOpCodeDelta}{-0.36}
\newcommand{\rsnOpCodeDeltaCI}{[-0.56, -0.16]}
\newcommand{\rsnOpFactsB}{0.38}
\newcommand{\rsnOpFactsDelta}{-0.19}
\newcommand{\rsnOpFactsDeltaCI}{[-0.39, -0.04]}
\newcommand{\rsnOpMathDelta}{-0.13}
\newcommand{\rsnOpMathDeltaCI}{[-0.27, 0.13]}
\newcommand{\rsnOpTruthDelta}{-0.17}
\newcommand{\rsnOpTruthDeltaCI}{[-0.28, -0.09]}
\newcommand{\sFourCoef}{$-0.06$}
\newcommand{\sFourCoefCI}{[-0.20, 0.06]}
\newcommand{\sFourMatchedGap}{-0.13}
\newcommand{\sFourMatchedGapCI}{[-0.33, 0.03]}
\newcommand{\sFourNPairs}{405}
\newcommand{\sFourRawCross}{0.60}
\newcommand{\sFourRawSame}{0.48}
\newcommand{\selfCotB}{0.80}
\newcommand{\selfCotBCI}{[0.60, 0.98]}
\newcommand{\selfCotBTerse}{1.00}
\newcommand{\selfCotDB}{-0.20}
\newcommand{\selfCotDBCI}{[-0.40, -0.02]}
\newcommand{\selfCotIdentB}{4}
\newcommand{\selfCotItems}{200}
\newcommand{\selfCotModels}{5}
\newcommand{\selfCotVotes}{4,000}
\newcommand{\sizeCodeDeltaCI}{[-0.105, -0.036]}
\newcommand{\sizeCodeDeltaTwentyFiveMinusThree}{-0.069}
\newcommand{\sizeCodeThreeCI}{[0.480, 0.571]}
\newcommand{\sizeCodeThreeLoss}{0.527}
\newcommand{\sizeCodeTwentyFiveBetterPct}{68}
\newcommand{\sizeCodeTwentyFiveCI}{[0.391, 0.530]}
\newcommand{\sizeCodeTwentyFiveLoss}{0.458}
\newcommand{\sizeCurveBootB}{1000}
\newcommand{\sizeCurveRosters}{500}
\newcommand{\sizeFactsDeltaCI}{[-0.073, -0.027]}
\newcommand{\sizeFactsDeltaTwentyFiveMinusThree}{-0.051}
\newcommand{\sizeFactsThreeCI}{[0.203, 0.295]}
\newcommand{\sizeFactsThreeLoss}{0.245}
\newcommand{\sizeFactsTwentyFiveBetterPct}{74}
\newcommand{\sizeFactsTwentyFiveCI}{[0.141, 0.245]}
\newcommand{\sizeFactsTwentyFiveLoss}{0.194}
\newcommand{\sizeMathDeltaCI}{[0.003, 0.048]}
\newcommand{\sizeMathDeltaTwentyFiveMinusThree}{+0.025}
\newcommand{\sizeMathThreeCI}{[0.865, 0.924]}
\newcommand{\sizeMathThreeLoss}{0.895}
\newcommand{\sizeMathTwentyFiveBetterPct}{44}
\newcommand{\sizeMathTwentyFiveCI}{[0.883, 0.955]}
\newcommand{\sizeMathTwentyFiveLoss}{0.921}
\newcommand{\sizeTruthDeltaCI}{[-0.090, -0.021]}
\newcommand{\sizeTruthDeltaTwentyFiveMinusThree}{-0.056}
\newcommand{\sizeTruthThreeCI}{[0.465, 0.578]}
\newcommand{\sizeTruthThreeLoss}{0.522}
\newcommand{\sizeTruthTwentyFiveBetterPct}{62}
\newcommand{\sizeTruthTwentyFiveCI}{[0.393, 0.540]}
\newcommand{\sizeTruthTwentyFiveLoss}{0.467}
\newcommand{\spFrFactsDelta}{-0.06}
\newcommand{\spFrFactsDeltaCI}{[-0.22, 0.08]}
\newcommand{\spFrTruthDelta}{-0.02}
\newcommand{\spFrTruthDeltaCI}{[-0.13, 0.23]}
\newcommand{\spOpCodeDelta}{-0.12}
\newcommand{\spOpCodeDeltaCI}{[-0.37, 0.09]}
\newcommand{\spOpFactsDelta}{-0.19}
\newcommand{\spOpFactsDeltaCI}{[-0.40, -0.07]}
\newcommand{\spOpMathDelta}{-0.08}
\newcommand{\spOpMathDeltaCI}{[-0.28, 0.13]}
\newcommand{\spOpTruthDelta}{-0.11}
\newcommand{\spOpTruthDeltaCI}{[-0.22, -0.03]}
\newcommand{\unparsedThinkPct}{0.9}
\newcommand{\SelectionIndependenceRelativeBenefitPct}{12.86}
\newcommand{\SelectionIndependenceRelativeBenefitCILowPct}{10.38}
\newcommand{\SelectionIndependenceRelativeBenefitCIHighPct}{14.91}
\newcommand{\SelectionFullIndependenceThresholdAgreementPct}{59.10}
\newcommand{\SelectionFullIndependenceThresholdAgreementCILowPct}{56.01}
\newcommand{\SelectionFullIndependenceThresholdAgreementCIHighPct}{63.81}
\newcommand{\SelectionArchivedCompositionRelativeBenefitPct}{18.96}
\newcommand{\SelectionPoolKindBalancedRelativeBenefitPct}{13.68}
\newcommand{\SelectionSelfOnlyRelativeBenefitPct}{2.04}
\newcommand{\SelectionCrossMixedRelativeBenefitPct}{19.91}
\newcommand{\SelectionFrechetSaturationPct}{18.48}
\newcommand{\SelectionUnbracketedClipPct}{17.09}
\newcommand{\SelectionPSDProjectionPct}{76.80}
\newcommand{\SelectionPSDProjectionMeanDistance}{1.075}
\newcommand{\SelectionPSDProjectionMaxDistance}{6.991}
\newcommand{\ParserAuditReviewedTotalN}{3149}
\newcommand{\ParserAuditAgreementRecordsN}{2754}
\newcommand{\ParserAuditAdjudicatedDisagreementsN}{395}
\newcommand{\ParserAuditPrimaryOverlayN}{2506}
\newcommand{\ParserAuditPrimaryOverlayApproveN}{735}
\newcommand{\ParserAuditPrimaryOverlayRejectN}{685}
\newcommand{\ParserAuditPrimaryOverlayMissingN}{1086}
\newcommand{\ParserAuditPrimaryWeightedPopulationN}{174384}
\newcommand{\ParserAuditPrimaryBinaryCoveragePct}{99.3772}
\newcommand{\ParserAuditPrimaryWeightedLabelDiscordancePct}{0.0866}
\newcommand{\ParserAuditUnresolvedCensusN}{1149}
\newcommand{\ParserAuditUnresolvedApproveN}{149}
\newcommand{\ParserAuditUnresolvedRejectN}{1}
\newcommand{\ParserAuditUnresolvedIndeterminateN}{999}
\newcommand{\ParserAuditAgreementOnlyAdditionalMissingN}{274}
\newcommand{\ParserAuditAgreementOnlyOverlayMissingN}{1360}
\newcommand{\SelectionForcedRejectPointRelativeBenefitPct}{15.70}
\newcommand{\kofn}{$k$-of-$n$}

\title{State-dependent error correlations shape voting thresholds in committees of AI agents}
\author{
Haifeng Li\\
\small School of Information, Central University of Finance and Economics\\
\small Beijing 100081, China\\
\small \texttt{mydlhf@cufe.edu.cn}; ORCID 0000-0001-6206-3662
\and
Mo Hai\thanks{Corresponding author: \texttt{haimo@cufe.edu.cn}}\\
\small School of Information, Central University of Finance and Economics\\
\small Beijing 100081, China\\
\small ORCID 0009-0005-1369-4341
}
\date{}

\begin{document}
\maketitle

\begin{abstract}
The aggregation benefit of a committee of artificial intelligence (AI) agents
comes from complementary information across members. Classical voting
guarantees assume independent errors. Language-model errors often co-occur on
the same cases. We combine Sah--Stiglitz screening with error dependence that
can differ between good and bad cases. In a homogeneous exchangeable Gaussian-copula model,
shared errors create a positive asymptotic error floor for majority voting and
can change the approval threshold that minimizes expected loss. We estimate a
heterogeneous extension from \nVotesBaseExact{} votes cast by \nModels{}
language models on four binary-screening benchmarks. Parameters estimated
from odd-indexed items predicted committee loss on even-indexed items. For the
sampled committee composition, the full-matrix dependence model increased
identity-line $R^2$ from \robustIndepRtwo{} under independence to
\robustFullRtwo{}. In a design-balanced analysis, cost-sensitive threshold
selection under independence reduced scaled loss from
\SelectionPrimaryMajorityScaledLoss{} for majority to
\SelectionPrimaryIndependenceScaledLoss{}. Modeling dependence reduced it
further to \SelectionPrimaryFullScaledLoss{}, an incremental improvement of
\SelectionPrimaryAbsoluteBenefitVsIndependenceScaledUnits{} units (95\%
bootstrap CI,
\SelectionPrimaryAbsoluteBenefitVsIndependenceCILowScaledUnits{}--\SelectionPrimaryAbsoluteBenefitVsIndependenceCIHighScaledUnits{}).
The overall reduction from majority was \SelectionPrimaryRelativeBenefitPct{}\%
(95\% bootstrap CI,
\SelectionPrimaryRelativeBenefitCILowPct{}--\SelectionPrimaryRelativeBenefitCIHighPct{}\%).
\end{abstract}

\noindent\textbf{Keywords:} artificial intelligence; collective
decision-making; correlated errors; ensemble methods; organizational design

\section*{Significance}
Organizations increasingly use panels of artificial intelligence systems to verify facts, review code, and screen decisions. Classical majority-voting guarantees assume independent errors. Language-model errors often co-occur on the same cases. In an exchangeable one-factor model, positive state-conditional dependence leaves nonzero large-committee error and can change the voting threshold that minimizes expected loss. A model calibrated on votes from 28 language models across four tasks predicts loss on held-out items and selects thresholds when false acceptance and false rejection have different costs. The results provide a practical way to design committees from measured error patterns.

\noindent
Organizations combine fallible judgments through approval chains, parallel
review, and voting. Sah and Stiglitz formalized the resulting trade-off.
Unanimity suppresses false acceptances and rejects more good proposals.
Permissive rules rescue good proposals and admit more bad ones
\cite{sah1985human,sah1986architecture,sah1988committees}. The best rule
depends on members' state-conditional error rates and on the relative cost of
the two mistakes.

Fixed language-model panels are increasingly used for fact checking, code
review, and model-as-a-judge evaluation. Other multi-agent systems use debate
and mixture-of-agents architectures
\cite{zheng2023judge,verga2024jury,du2024debate,wang2024moa,wu2023autogen,hong2024metagpt,chan2024chateval}.
Condorcet-style reliability guarantees for majority voting and
self-consistency rely on independent-vote logic
\cite{condorcet1785,wang2023selfconsistency,li2024moreagents}.
Empirical studies find above-chance co-occurrence in model errors
\cite{kim2025correlated,goel2025great}.
Dependence changes both the value of adding judges and the threshold that
minimizes cost-weighted loss.

Jury theory characterizes dependence and optimal rules for exchangeable
jurors, and independence-based theory characterizes optimal aggregation for
heterogeneous jurors
\cite{ladha1992,boland1989,berend2007,kaniovski2011,nitzan1982}. Ensemble
learning links diversity to predictive performance, and recent studies measure
shared errors and roster-dependent performance in language-model systems
\cite{breiman2001,kuncheva2003,kim2025correlated,goel2025great,kim2025scaling}.
Our screening formulation combines three features: members differ in their
state-conditional approval rates, dependence can differ between good and bad
cases, and the threshold reflects the prevalence and consequences of both
errors.

Here we treat state-conditional dependence as a measurable design input. A
homogeneous one-factor Gaussian model shows how dependence changes the
protection supplied by threshold rules, the limiting reliability of majority
voting, and the loss-minimizing threshold. A heterogeneous Gaussian-copula
model then predicts the loss of every \kofn{} rule from member-specific
approval probabilities and two state-specific correlation matrices.

We estimate these quantities on odd-indexed items and evaluate threshold choices
on even-indexed items. Comparisons with majority voting and an otherwise
matched independence model separate the gain from cost-sensitive threshold
choice from the additional gain from dependence modeling.

\section*{Results}
\subsection*{State-dependent error correlation and voting thresholds}

\noindent
\textbf{Setup.} A proposal is either good ($G$, prior $0<\pi_G<1$) or bad
($B$, prior $\pi_B=1-\pi_G>0$), and the organization accepts or rejects it.
Each of $n$ agents casts an approve or reject vote $V_i$. The
state-conditional approval probability of each agent is
$p_\theta = \Pr(V_i = 1 \mid \theta)$, with
$0<p_B<p_G<1$ for informative agents. Dependence enters through a single
latent factor. Agent $i$'s score in state $\theta$ is $Z_i = \sqrt{\rho_\theta}\, W +
\sqrt{1-\rho_\theta}\,\varepsilon_i$ with $W, \varepsilon_1, \dots,
\varepsilon_n$ independent standard normal,
and $V_i = \mathbf 1\{Z_i \le \Phi^{-1}(p_\theta)\}$ (Fig.~1A). The common
factor $W$ represents a shared component of judgment, and
$\rho_\theta \in [0,1)$ is the \emph{latent error correlation} in state
$\theta$. A decision rule is a \kofn{} threshold that accepts iff at least $k$
agents approve (unanimity at $k{=}n$, permissive acceptance at $k{=}1$, majority at
$k{=}\lceil (n{+}1)/2 \rceil$). Design quality is normalized expected loss
$R_{n,k} = (1 - A_{n,k}(p_G,\rho_G)) + \kappa\, A_{n,k}(p_B,\rho_B)$, where the
acceptance function $A_{n,k}(p,\rho) = \Pr(\sum_i V_i \ge k)$ and
$\kappa=(\pi_Bc_{\mathrm I})/(\pi_Gc_{\mathrm{II}})$ is the cost-weighted
prior odds of the bad state.

Because the common factor shifts every agent equally, the \kofn{} rule
accepts exactly when the $k$-th smallest latent score falls below the approval
cutoff:
\begin{equation}
A_{n,k}(p,\rho) \;=\; \Pr\!\Bigl(\, \underbrace{\sqrt{\rho}\, W}_{\text{
undiversifiable}} \;+\; \underbrace{\sqrt{1-\rho}\;
\varepsilon_{(k)}}_{\text{diversifiable}} \;\le\; \Phi^{-1}(p) \Bigr),
\label{eq:order}
\end{equation}
with $\varepsilon_{(k)}$ the $k$-th order statistic of $n$ independent
standard normals (Appendix, Lemma S2). Equation \eqref{eq:order} separates the
threshold-dependent order statistic of the agent-specific component from the
shared component $\sqrt{\rho}\,W$, which is unaffected by the threshold.

\medskip
\noindent
\textbf{Result 1. Correlation weakens the protection of extreme rules.}
Raising $\rho$ produces a mean-preserving spread of the conditional approval
rate (Appendix, Lemma S3). This yields global comparisons for the two extreme rules.
Writing $q_\theta(w)=\Pr(V_i=1\mid W=w,\theta)$, for $n\ge2$ hierarchies
filter bad proposals. Their false-acceptance rate
$A_{n,n}(p_B,\rho_B) = \mathbb E[q_B(W)^n] \ge p_B^n$ rises strictly with
$\rho_B$ above its independent value $p_B^n$ (Jensen; equality iff
$\rho_B=0$). Polyarchies rescue good proposals. Their rescue rate falls
strictly with $\rho_G$. For odd $n\ge3$, the majority-acceptance curve is
convex below one half and concave above it. At independence, the two partial
right derivatives of loss with respect to $\rho_G$ and $\rho_B$ are positive
when $p_B < \tfrac12 < p_G$. The derivative is therefore positive along the
common-correlation path $\rho_G=\rho_B=\rho$. On a prespecified finite grid with
$\kappa \in \{0.25,1,4\}$ along the common-correlation path
$\rho_G=\rho_B=\rho\le 0.9$ (2{,}700 parameter curves), majority loss was
also nondecreasing across the grid (Appendix, \S S3 and \S S8). Correlation thus
erodes the state-specific protection supplied by each canonical rule
(Fig.~1B).

\medskip
\noindent
\textbf{Result 2. Correlated majorities have a nonzero error limit.}
Independence yields exponentially vanishing committee error under Condorcet's
logic. Positive correlation yields convergence to the Vasicek limit. For
either state, let $p>1/2$ denote the
probability that one agent makes the correct state-contingent decision
($p=p_G$ in state $G$ and $p=1-p_B$ in state $B$), and let $\rho$ be that
state's latent correlation, with $0<\rho<1$. For majority rule,
\begin{equation}
\varepsilon^{\infty}(p, \rho) \;=\; \Phi\!\left( - \frac{\Phi^{-1}(p)}
{\sqrt{\rho}} \right) \;>\; 0,
\label{eq:floor}
\end{equation}
(Appendix, Corollary S8). The \emph{infinite-majority residual fraction} of a single
agent's error is therefore $\mathcal F =
\Phi(-\Phi^{-1}(p)/\sqrt\rho) / \Phi(-\Phi^{-1}(p))$ (Fig.~1C). For an agent
with accuracy $p = 0.8$, the 5th--95th-percentile correlations measured for
open-weight cross-family pairs leave \floorFracMinPct{}--\floorFracMaxPct{}\%
of single-agent error in the infinite-majority limit of this homogeneous
reference model. Positive correlation in both states also creates a positive
asymptotic Bayes-risk floor among all rules that use only the vote vector
(Appendix, Prop.~S10).

\medskip
\noindent
\textbf{Result 3. Near-perfect correlation collapses a fixed committee.}
As $\rho \to 1$, $A_{n,k}(p,\rho) \to p$ for every fixed $n$ and every $k$
(Appendix, Theorem S6), so every fixed threshold rule converges to single-agent
behavior. This conclusion extends to all vote-only aggregation rules when
committee size is bounded and $(\rho_G,\rho_B)\to(1,1)$ (Appendix, Theorem S15).
This fixed-size correlation limit differs from the large-committee limit in
Result 2, which holds correlation below one and sends $n\to\infty$.

\medskip
\noindent
\textbf{Result 4. Cost determines the optimal threshold.}
For a fixed committee, the two state-conditional vote-count distributions
determine the loss of every \kofn{} rule. Direct enumeration therefore gives
the complete optimal threshold set for both homogeneous and heterogeneous
rosters. The smallest and largest optimal thresholds are nondecreasing in
$\kappa$, so costlier false acceptances call for stricter rules
(Appendix, Prop.~S13). When a homogeneous committee has equal state correlations,
its vote-count likelihood ratio is monotone and the optimal set begins where
that ratio weakly crosses $\kappa$ (Appendix, Theorem S12). Unequal state
correlations can break this property, so the empirical analysis enumerates
all candidate losses. Figure~1D shows how correlation changes the
loss-minimizing threshold regions. When $\rho_B > \rho_G$, unanimous approval
can signal the correlated bad state and provide weaker evidence of quality
(Appendix, Remark S11; finite-grid failures in Appendix, \S S8, example in Appendix, Fig.~S1)
\cite{gunn2016}.

\begin{figure}[t!]
\centering
\includegraphics[width=\textwidth]{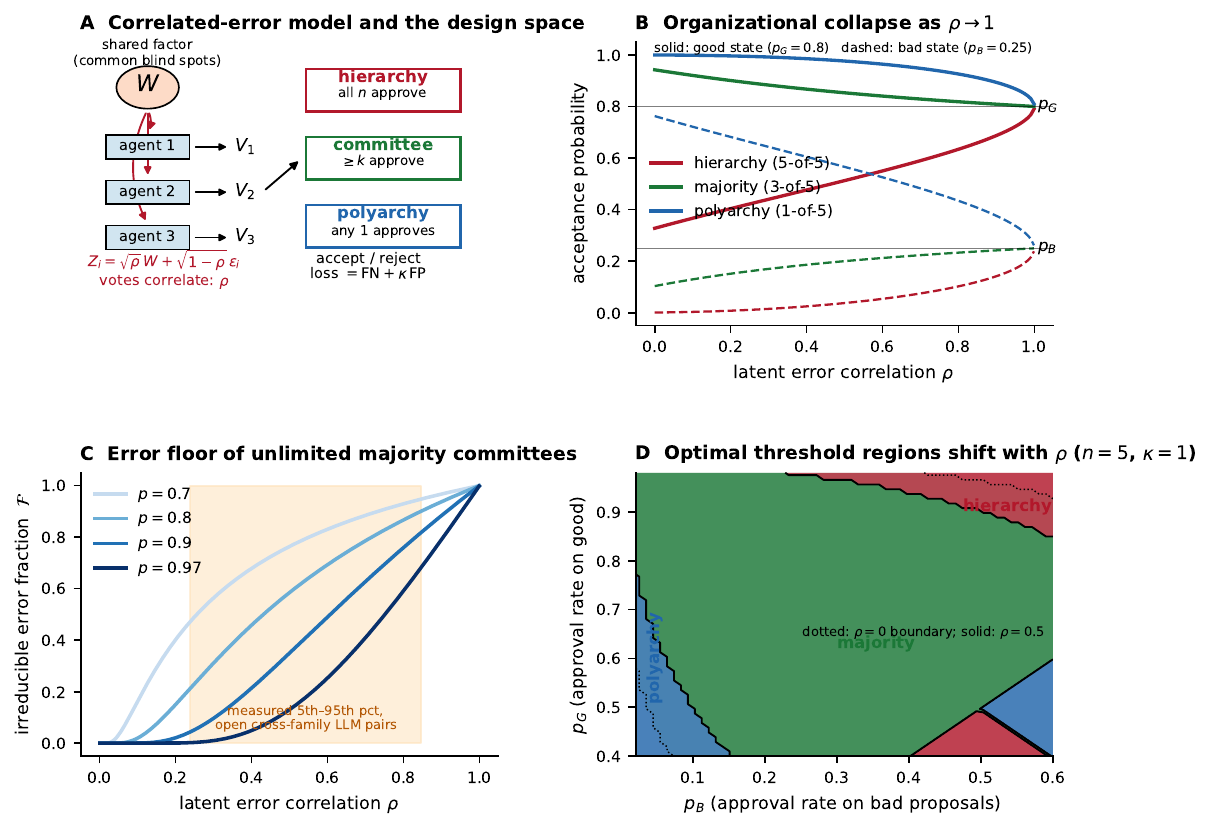}
\caption{\textbf{State-dependent error correlation and committee thresholds.}
(\textbf{A}) The screening model. Agents share a latent factor $W$, and their
latent scores have correlation $\rho$. A \kofn{} threshold maps votes to a
decision evaluated
by cost-weighted loss. (\textbf{B}) Acceptance probabilities versus $\rho$
for hierarchy, majority, and polyarchy ($n=5$) in the good state
($p_G = 0.8$, solid) and bad state ($p_B = 0.25$, dashed); each
rule's signature protection erodes with $\rho$, and all rules
collapse to single-agent behavior ($A \to p$) as $\rho \to 1$ (Results 1 and
3). (\textbf{C}) The infinite-majority residual fraction $\mathcal F$ of
single-agent error under the homogeneous reference model
(Eq.~\ref{eq:floor}) for agents of accuracy $p$. The shaded
band marks the 5th--95th percentile of latent correlations measured for
identifiable open-weight cross-family large-language-model pairs, based on 726 state-specific
pair--domain correlation estimates. (\textbf{D})
Loss-minimizing voting-threshold phase diagram in $(p_B, p_G)$ at
$\kappa = 1$, $n=5$. Colors and solid boundaries show all five thresholds
at $\rho = 0.5$, and dotted boundaries show $\rho = 0$. The gray region
$p_G\leq p_B$ marks noninformative agents.}
\label{fig:theory}
\end{figure}

\subsection*{Measured correlation structure of machine judgments}

\noindent
The theory raises two empirical questions. How large and state dependent are
the correlations among model judgments, and does modeling them improve
held-out loss prediction and threshold choice?

We measured the theory's inputs in four approve or reject screening domains
with known labels. The domains were
\emph{fact verification} (does evidence support the claim?
\cite{schuster2021vitaminc}), \emph{mathematical solution verification} (is a
solution correct? half carry a planted arithmetic corruption
\cite{cobbe2021gsm8k}), \emph{code review} (is a program correct? 44\% carry
a mutation verified to fail the official tests \cite{austin2021mbpp}), and
\emph{truthfulness screening} (true answer or plausible misconception?
\cite{lin2022truthfulqa}). The roster contains \nModelsOpen{} open-weight
instruction-tuned models from \nFamiliesOpen{} families and
\nModelsFrontier{} configured API models (full roster in Appendix, \S S9). Each model
evaluated every item once at temperature 0 and three times at temperature 0.7,
yielding \nVotesBaseExact{} votes. In the committee analyses, an agent is a
model--sampling-run pair. Cross-model committees use one run per model;
same-model committees use distinct temperature samples. From these votes we
estimated each agent's $(\hat p_G, \hat p_B)$ and every pair's latent
correlation $(\hat\rho_G, \hat\rho_B)$ by tetrachoric inversion.
Validation parameters came from odd-indexed items and losses from even-indexed
items. Descriptive correlation maps used all items, with bootstrap 95\%
intervals on key estimates (Appendix, \S S9).

Dependence was substantial wherever the binary marginals permitted stable
estimation (Fig.~2). Same-model temperature resampling had near-unit latent
correlation (median
$\hat\rho=\rhoSelfMedian$; mean \rhoSelfTypical{} across models and domains),
and its majority-of-three gain was correspondingly small (mean loss reduction
\selfGainLoss). Open-weight cross-family pairs had average latent correlations
$\hat\rho_G\approx\rhoCrossG$ and $\hat\rho_B\approx\rhoCrossB$
(Fig.~2A--C), below the same-model values. These summaries cover identifiable
pairs; nearly constant voting patterns fall outside the tetrachoric
identifiability window.

After adjustment for capability and domain, the estimated association between
same-family membership and bad-state dependence had a 95\% interval that
included zero (Appendix, \S S10). Measured marginals and pairwise dependence
therefore provide the direct inputs for ranking pair complementarity in this
roster.

\begin{figure}[t!]
\centering
\includegraphics[width=\textwidth]{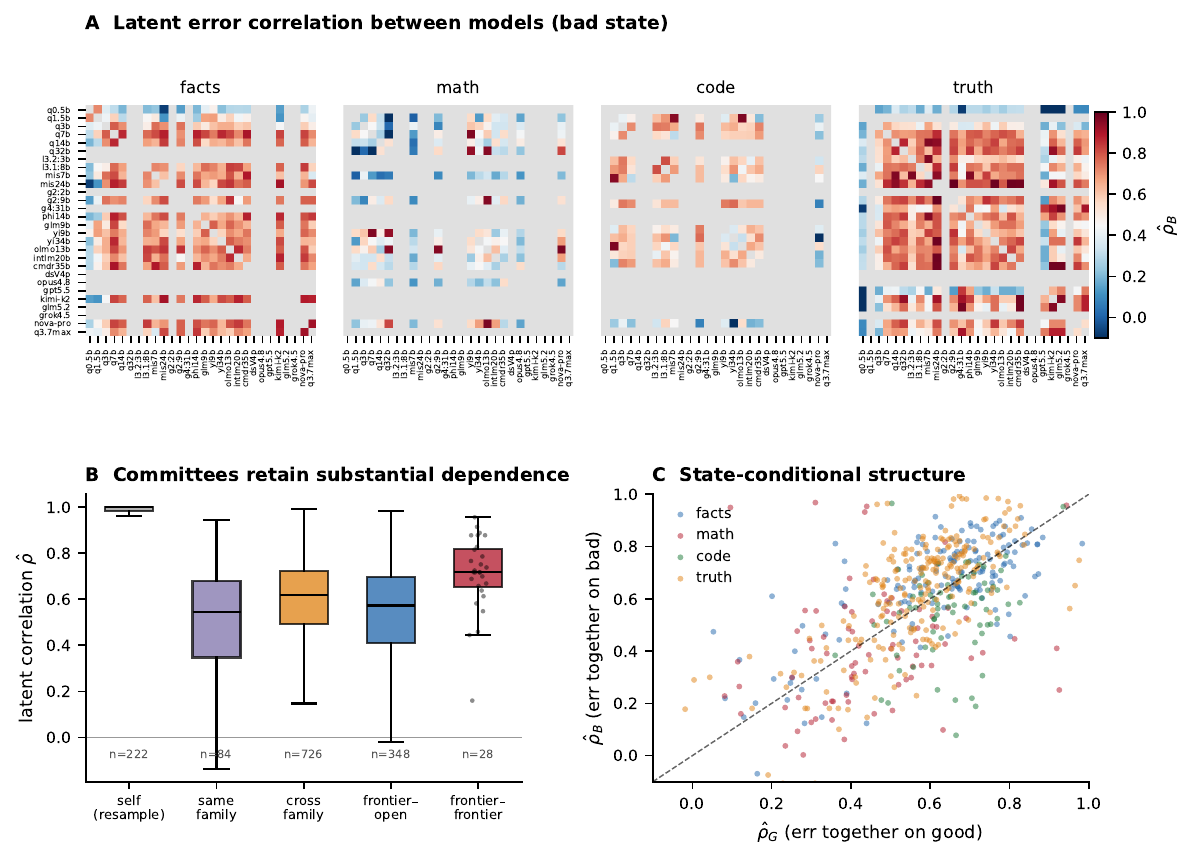}
\caption{\textbf{Measured correlation structure of machine judgments.}
(\textbf{A}) Latent error-correlation matrices $\hat\Sigma_B$ (bad state;
tetrachoric estimates between the first temperature sample of each judge,
all items; Phi-4 is excluded in code review, Appendix, \S S9), by domain; gray
cells mark pairs in which a
member approves or rejects nearly everything, leaving the tetrachoric
correlation unidentifiable, and concentrate in the near-chance
verification domains. (\textbf{B}) Distribution of pairwise $\hat\rho$, pooling
$\hat\rho_G$ and $\hat\rho_B$ over identifiable pairs ($n$ = pair--state
values, below each box). Boxes span the first and third quartiles, center lines
mark medians, and whiskers end at 1.5 times the interquartile range.
Resampling one model yields
near-clones ($\hat\rho \approx 1$), and committees of different models,
within or across families and access modes, retain substantial positive
latent correlation.
(\textbf{C}) State-conditional structure: $\hat\rho_G$ versus
$\hat\rho_B$ for 540 identifiable cross-family pair--domain observations; dependence can
differ between good and bad cases.}
\label{fig:measurement}
\end{figure}

\subsection*{Validation of committee-loss prediction on held-out items}

\noindent
For the sampled validation composition, committee draws span
$n\in\{3,5,7,9\}$, four domains, several roster constructions, and five
values $\kappa\in\{0.25,0.5,1,2,4\}$. Parameters are estimated on odd-indexed
items and every $k$-of-$n$ loss is evaluated on even-indexed items. The
\robustAblationN{} threshold--cost rows share committee draws. We use this
sampled composition for descriptive calibration and a separate balanced
analysis for threshold-selection inference.

Under that composition, the heterogeneous full-matrix Gaussian-copula
model tracks loss on held-out items closely (identity-line
$R^2=\robustFullRtwo$, RMSE \robustFullRMSE; Fig.~3A). A nonnegative
exchangeable one-factor projection gives
$R^2=\robustExchangeRtwo{}$ (RMSE \robustExchangeRMSE). The independence
model gives $R^2=\robustIndepRtwo{}$ (RMSE \robustIndepRMSE;
Fig.~3B--C). Committee-cluster intervals and omission ranges are reported in
Appendix, Table~S4.

The primary comparison assigns equal weight to 20
domain--pool-kind--size cells. In each of \SelectionBootstrapValid{} replicates,
parameters and thresholds are re-estimated after resampling the estimation
items, evaluation items, and whole committees within their design cells.

After multiplying loss by 100, the design-balanced means were
\SelectionPrimaryMajorityScaledLoss{} for majority,
\SelectionPrimaryIndependenceScaledLoss{} for the independence-based
threshold, and \SelectionPrimaryFullScaledLoss{} for the dependence-aware
threshold (Appendix, Table~S6). The independence-based threshold reduced loss by
\SelectionIndependenceBenefitVsMajorityScaledUnits{} units on this scale
(95\% bootstrap CI,
\SelectionIndependenceBenefitVsMajorityCILowScaledUnits{}--\SelectionIndependenceBenefitVsMajorityCIHighScaledUnits{}).
Adding the fitted dependence structure reduced loss by a further
\SelectionPrimaryAbsoluteBenefitVsIndependenceScaledUnits{} units
(95\% bootstrap CI,
\SelectionPrimaryAbsoluteBenefitVsIndependenceCILowScaledUnits{}--\SelectionPrimaryAbsoluteBenefitVsIndependenceCIHighScaledUnits{}).
These components accounted for \SelectionThresholdOnlySharePct{}\% and
\SelectionDependenceSharePct{}\% of the
\SelectionPrimaryAbsoluteBenefitVsMajorityScaledUnits{}-unit total reduction
(95\% bootstrap CI,
\SelectionPrimaryAbsoluteBenefitVsMajorityCILowScaledUnits{}--\SelectionPrimaryAbsoluteBenefitVsMajorityCIHighScaledUnits{})
on this scale. The total was equivalent to \SelectionPrimaryRelativeBenefitPct{}\%
(95\% bootstrap CI,
\SelectionPrimaryRelativeBenefitCILowPct{}--\SelectionPrimaryRelativeBenefitCIHighPct{}\%).
The two models selected different thresholds in
\SelectionFullIndependenceThresholdDisagreementPct{}\% of weighted
committee--$\kappa$ cases (95\% bootstrap CI,
\SelectionFullIndependenceThresholdDisagreementCILowPct{}--\SelectionFullIndependenceThresholdDisagreementCIHighPct{}\%).
Selected thresholds became stricter as $\kappa$ rose and more
permissive as it fell, consistent with the comparative static (Fig.~3E).
The intervals describe the observed roster, benchmark domains, and committee
design.

\begin{figure}[t!]
\centering
\includegraphics[width=\textwidth]{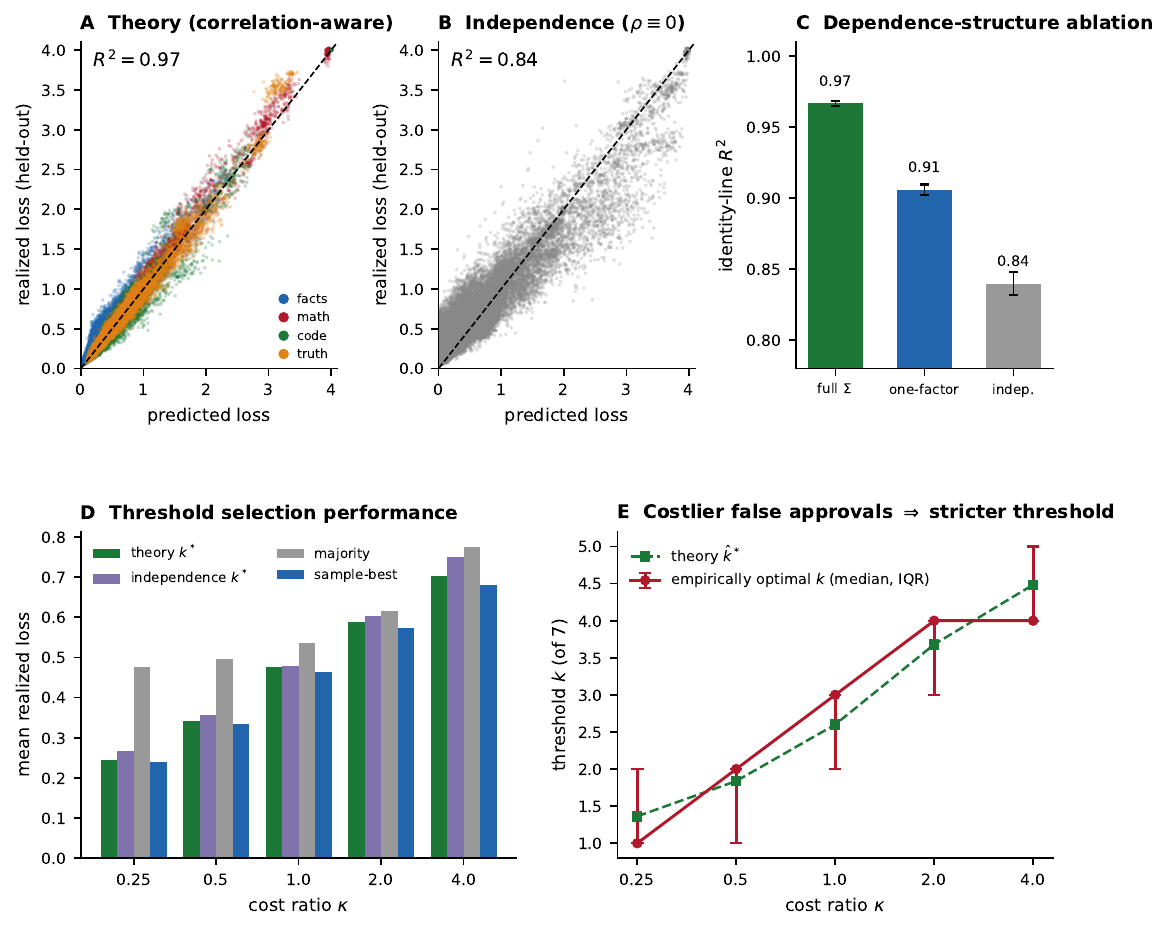}
\caption{\textbf{Validation of committee-loss prediction on held-out items.}
(\textbf{A}) Heterogeneous full-matrix dependence predictions versus realized
loss for all \kofn{} thresholds under the archived committee composition
(\robustAblationN{} threshold--cost rows). Parameters
are estimated from odd-indexed items and evaluated on even-indexed items.
(\textbf{B}) Predictions after forcing independence.
(\textbf{C}) Descriptive identity-line $R^2$ for the full-matrix model,
a nonnegative exchangeable one-factor projection, and independence. The 95\%
intervals group all threshold--cost rows from the same committee in each
bootstrap cluster ($B=\robustBootB{}$).
(\textbf{D}) Sampled-composition point losses after selecting a threshold
with the dependence-aware model, independence model, or majority default,
with the evaluation-sample oracle shown for reference. Design-balanced selection
inference is reported in Appendix, Table~S6.
(\textbf{E}) Empirical optimal thresholds for $n=7$ and corresponding
dependence-aware model predictions shift upward as cost-weighted bad-state
odds $\kappa$ rise.}
\label{fig:validation}
\end{figure}

\section*{Discussion}

\noindent
The loss-minimizing threshold is jointly determined by state prevalence,
relative error costs, member-specific approval rates, and state-conditional
dependence. When $\kappa$ is uncertain,
intersections of the candidate loss lines identify thresholds that remain
optimal over a plausible interval (Appendix, \S S1).

The homogeneous model identifies how shared errors limit aggregation. The
heterogeneous model evaluates a fixed roster from member-specific marginals
and full state-specific dependence matrices. On held-out items, retaining this
structure improved both loss prediction and threshold choice.

The empirical evidence covers four binary-screening benchmarks and the
observed model roster, with held-out items drawn from those same benchmarks.
The fitted tetrachoric matrices often require boundary handling and
positive-semidefinite projection (Appendix, Table~S6). Correlation is also unidentified
for nearly constant judges. The resulting estimates are regularized
descriptions of this roster.

Parser uncertainty was evaluated with reviewed-record overlays. The relative
loss reduction ranged from \SelectionAgreementOnlyPointRelativeBenefitPct{}\%
when only records with agreeing review passes were retained to
\SelectionListwisePointRelativeBenefitPct{}\% under listwise deletion of
indeterminate reviewed labels (Appendix, \S S9 and Table~S6).

Cost-sensitive threshold choice supplied most of the observed improvement over
majority. Dependence modeling added a smaller positive loss reduction, changed
the selected threshold in \SelectionFullIndependenceThresholdDisagreementPct{}\%
of weighted cases, and improved loss prediction. State-dependent error
dependence thus links screening theory to a practical design rule whose
choices can be tested on held-out items.

\section*{Materials and Methods}
\subsection*{Theory and numerical checks}
The exchangeable model conditions on one Gaussian common factor and evaluates
the resulting mixed-binomial acceptance probability by quadrature. Analytic
results and proofs are given in Appendix, sections S1--S7; section S8
reports the numerical checks of formulas, endpoints, stated monotonicities,
and prespecified finite grids.

\subsection*{Screening domains, judges, and voting protocol}
The four domains were fact verification (400 VitaminC claim--evidence pairs),
mathematical solution verification (400 GSM8K problems, half with one planted
arithmetic corruption), code review (357 MBPP programs, 157 with a mutation
verified to fail an official test), and truthfulness screening (400
TruthfulQA question--answer pairs). Every item had a known good or bad state,
and judges returned one binary approve or reject decision. The roster
contained \nModelsOpen{} open-weight instruction-tuned models and
\nModelsFrontier{} models accessed through application programming interfaces.
Open-weight models were served with \texttt{{Q4\_K\_M}} quantization on one NVIDIA A100
80-GB graphics processor. Each primary judge voted once at temperature zero
and in three temperature-0.7 samples. For analysis, each
model--sampling-run pair is one agent. The archived corpus contains
\nVotesBaseExact{} primary-protocol votes. Model identities, prompts, sampling
settings, item construction, archive fields, and provider routes are reported
in Appendix, section S9.

\subsection*{Estimation and held-out validation}
Within each domain, odd-indexed items formed the estimation half and
even-indexed items the evaluation half. We estimated each judge's
state-conditional approval probabilities and pairwise tetrachoric latent
correlations. Pairwise joint probabilities were constrained by the binary
Fr\'echet bounds before tetrachoric inversion, and each committee correlation
submatrix was projected to the positive-semidefinite cone when required. For
each fitted committee, Gaussian-copula simulation predicted the loss of every
$k$-of-$n$ threshold at
$\kappa\in\{0.25,0.5,1,2,4\}$. The independence comparator used the identical
procedure with the correlation matrices fixed to identity. Realized losses
use votes from held-out items. Sampled-composition identity-line
$R^2$ and root-mean-square error are descriptive calibration summaries because
threshold--cost rows from the same committee share data.
Identity-line $R^2$ fixes the slope at one and the intercept at zero and is
computed as
$1-\sum_j(y_j-\hat y_j)^2/\sum_j(y_j-\bar y)^2$.

\subsection*{Primary selection inference}
The primary estimand assigns equal weight to 20 prespecified
domain--pool-kind--committee-size cells: self-sampled three-member committees;
cross-model committees of sizes 3, 5, and 7; and mixed nine-member committees
in each domain. Each of \SelectionBootstrapValid{} valid bootstrap replicates
independently resampled estimation and evaluation items within domain and
state, re-estimated all marginals and correlations, reselected the smallest
loss-minimizing threshold for the full-matrix dependence model and the
independence model, and resampled whole committee clusters within design
cells. Predictions for all
candidate thresholds shared \SelectionMonteCarloDraws{} antithetic common random
draws within each refit. The intervals are conditional on the observed model
roster, benchmark domains, and committee design. Reviewed-record parser
sensitivities and reasoning-protocol comparisons are described in SI
Appendix, sections S9--S10.

\subsection*{Use of generative tools}
OpenAI Codex (GPT-5, accessed July 26, 2026) assisted language editing and
LaTeX consistency checks.
Separate language-model reviews were used in the reviewed-record parser
sensitivity analysis, followed by deterministic rules and model-assisted
adjudication. The authors verified the final text, derivations, citations, and
analysis outputs.

\section*{Data Availability}
The archived vote corpus and source manifests are available in
Zenodo release v7.3.1 \cite{li2026replication}. The concept DOI
\url{https://doi.org/10.5281/zenodo.21303076} indexes the release series.

\section*{Acknowledgments}
The authors received no funding for this work.

\section*{Author Contributions}
H.L. designed and performed the experiments and analyses
and wrote the first draft; M.H. conceived and supervised the research,
developed the methodology, and revised the paper.

\section*{Competing Interests}
The authors declare no competing interest.

\clearpage
\appendix
\setcounter{section}{0}
\setcounter{equation}{0}
\setcounter{figure}{0}
\setcounter{table}{0}
\setcounter{theorem}{0}
\renewcommand{\thesection}{S\arabic{section}}
\renewcommand{\theequation}{S\arabic{equation}}
\renewcommand{\thefigure}{S\arabic{figure}}
\renewcommand{\thetable}{S\arabic{table}}
\renewcommand{\thetheorem}{S\arabic{theorem}}

\section*{Appendix: Supporting Information}

\section{Model and definitions}
\label{sm:model}

\subsection{Screening problem}
An organization must decide whether to \emph{accept} or \emph{reject} a proposal
(a claim to certify, a solution to approve, a piece of code to merge, an answer
to release). The proposal is in one of two latent states,
$\theta \in \{G, B\}$. Good proposals have prior
$0<\pi_G<1$ and should be accepted; bad proposals have prior
$\pi_B = 1-\pi_G>0$ and should be rejected. Wrongly
accepting a bad proposal (a type-I organizational error) costs $c_{\mathrm{I}}>0$;
wrongly rejecting a good one (type-II) costs $c_{\mathrm{II}}>0$. This is the
canonical screening problem of Sah and Stiglitz; our interest is in organizations
whose members are artificial intelligence systems.

\subsection{Agents and correlated errors}
The organization employs $n$ agents. Agent $i$ casts a vote
$V_i \in \{0,1\}$ ($1=$ approve). Its state-conditional approval probability is
$p_\theta = \Pr(V_i = 1 \mid \theta)$; throughout,
$0<p_B<p_G<1$. To model dependence we use a one-factor latent-threshold
(probit) specification, standard in the analysis of correlated binary
events; it is the exchangeable Gaussian copula, known in credit risk as
the Vasicek model:
\begin{equation}
\label{eq:factor}
Z_i^{(\theta)} \;=\; \sqrt{\rho_\theta}\, W \;+\; \sqrt{1-\rho_\theta}\,
\varepsilon_i, \qquad W,\ \varepsilon_1,\dots,\varepsilon_n
\ \overset{\text{iid}}{\sim}\ \mathcal N(0,1),
\end{equation}
\begin{equation}
V_i \;=\; \mathbf 1\!\left\{ Z_i^{(\theta)} \le c_\theta \right\},
\qquad c_\theta = \Phi^{-1}(p_\theta),
\end{equation}
where $\Phi$ is the standard normal CDF and $\rho_\theta \in [0,1)$ is the
\emph{latent error correlation} in state $\theta$. The common factor $W$
represents a shared latent component of the votes, and $\varepsilon_i$
represents agent-specific variation. Conditional on $W=w$ the votes are iid
Bernoulli with
\begin{equation}
\label{eq:qw}
q_\theta(w) \;=\; \Phi\!\left( \frac{c_\theta - \sqrt{\rho_\theta}\,
w}{\sqrt{1-\rho_\theta}} \right),
\qquad \mathbb E\left[q_\theta(W)\right] = p_\theta .
\end{equation}
The manifest (Pearson) correlation between two votes,
$r_\theta = \operatorname{corr}(V_i, V_j)$, is related to the latent
correlation by
$r_\theta = [\Phi_2(c_\theta, c_\theta; \rho_\theta) - p_\theta^2] /
[p_\theta (1 - p_\theta)]$, where $\Phi_2(\cdot,\cdot;\rho)$ is the bivariate
normal CDF; estimating $\rho_\theta$ from binary votes by inverting this map is
the classical tetrachoric correlation. We allow $\rho_G \neq \rho_B$
throughout; agents may err together more on bad proposals than good ones.
Section~\ref{sm:hetero} treats heterogeneous marginals $p_{\theta,i}$ and a
general latent correlation matrix $\Sigma_\theta$. The exchangeable case
provides the analytic reference model.

\subsection{Architectures}
Following Sah--Stiglitz, an \emph{architecture} is a rule mapping votes to an
organizational decision. The classical design space is the family of
$k$-of-$n$ threshold rules: accept iff $S \ge k$, where
$S=\sum_{i=1}^n V_i$. It contains the three canonical organizational forms:
\begin{center}
\begin{tabular}{lll}
$k=n$ & \textbf{hierarchy} & series screening; unanimity; \\
$k=1$ & \textbf{polyarchy} & parallel screening; one approval suffices; \\
$k=\lceil (n{+}1)/2 \rceil$ & \textbf{majority committee} & simple majority.
\end{tabular}
\end{center}
The \emph{acceptance function} of a rule is
$A_{n,k}(p,\rho) = \Pr(S \ge k)$ evaluated at the state-specific pair
$(p_\theta, \rho_\theta)$. The organization's normalized expected loss is
\begin{equation}
\label{eq:loss}
R_{n,k} \;=\; \underbrace{\bigl(1 - A_{n,k}(p_G,\rho_G)\bigr)}_{\text{missed
good proposals}} \;+\; \kappa\,
\underbrace{A_{n,k}(p_B,\rho_B)}_{\text{accepted bad proposals}},
\qquad
\kappa \;=\; \frac{\pi_B\, c_{\mathrm{I}}}{\pi_G\, c_{\mathrm{II}}},
\end{equation}
Thus $\kappa$ combines target-state prevalence with the relative consequences
of false acceptance and false rejection. Screening is conservative when
$\kappa \gg 1$, as at deployment gates and safety reviews where bad
acceptances are expensive, and liberal when $\kappa \ll 1$, as in idea triage
and anomaly flagging where missed positives are expensive.

\paragraph{Calibration of $\kappa$ for deployment.}
The target decision stream supplies the state odds $\pi_B/\pi_G$, estimated
from adjudicated historical cases or a calibrated prevalence model. The
consequence ratio $c_{\mathrm I}/c_{\mathrm{II}}$ expresses false acceptance
and false rejection in a common decision-relevant unit. Their product gives
the deployment value of $\kappa$. When either component is uncertain, a
plausible interval $[\kappa_L,\kappa_U]$ induces threshold-optimality regions
\[
\mathcal I_k=\{\kappa>0:R_{n,k}(\kappa)\le R_{n,j}(\kappa)
\ \text{for every }j=1,\ldots,n\}.
\]
Because $R_{n,k}(\kappa)$ is affine in $\kappa$, these regions follow from
pairwise intersections of the candidate loss lines; Proposition~\ref{prop:topkis}
orders them from permissive to strict thresholds. Intersections with
$[\kappa_L,\kappa_U]$ give the selected threshold and its stability over the
deployment-relevant range.

All interior results take
$n\in\mathbb N$, $1\leq k\leq n$, $0\leq\rho_G,\rho_B<1$, and
$0<\kappa<\infty$; boundary results state their limits explicitly.

\section{Exact acceptance functions}
\label{sm:exact}

Write
$B_{n,k}(x)=\Pr\{\operatorname{Bin}(n,x)\ge k\}$ for the binomial upper-tail
probability.

\begin{lemma}[Mixed-binomial representation]
\label{lem:rep}
Under model \eqref{eq:factor},
\begin{equation}
\label{eq:Ank}
A_{n,k}(p,\rho) \;=\; \int_{-\infty}^{\infty}
\Pr\bigl(\mathrm{Bin}(n, q(w)) \ge k\bigr)\, \varphi(w)\, dw ,
\end{equation}
with $q(w)$ as in \eqref{eq:qw} and $\varphi$ the standard normal density.
In particular $A_{n,k}(p,0) = \Pr(\mathrm{Bin}(n,p) \ge k)$.
\end{lemma}

\begin{proof}
Immediate from conditional independence of the $V_i$ given $W$ and Fubini's
theorem.
\end{proof}

\begin{lemma}[Order-statistic representation]
\label{lem:order}
Let $\varepsilon_{(k)}$ denote the $k$-th smallest of
$\varepsilon_1,\dots,\varepsilon_n$. The $k$-of-$n$ rule accepts iff
$Z_{(k)} \le c$, and
\begin{equation}
\label{eq:orderstat}
A_{n,k}(p,\rho) \;=\; \Pr\Bigl( \sqrt{\rho}\, W + \sqrt{1-\rho}\,
\varepsilon_{(k)} \;\le\; c \Bigr).
\end{equation}
\end{lemma}

\begin{proof}
At least $k$ of the $Z_i$ fall below $c$ iff the $k$-th smallest does. Since
$Z_i = \sqrt\rho\, W + \sqrt{1-\rho}\,\varepsilon_i$ and $W$ is common to all
agents, the order statistics of $(Z_i)$ are
$Z_{(k)} = \sqrt\rho\, W + \sqrt{1-\rho}\,\varepsilon_{(k)}$.
\end{proof}

Lemma~\ref{lem:order} decomposes the organization's decision statistic into an
\emph{undiversifiable} systematic term $\sqrt{\rho}\,W$ of variance $\rho$,
whose variance remains $\rho$ for every committee size, and a
\emph{diversifiable} idiosyncratic
term $\sqrt{1-\rho}\,\varepsilon_{(k)}$ whose randomness vanishes as the
organization grows ($\varepsilon_{(\lceil \tau n\rceil)} \to \Phi^{-1}(\tau)$
almost surely). Organizational design affects the second term; correlation
fixes the size of the first.

\section{Comparative statics in the correlation}
\label{sm:compstat}

\begin{lemma}[Correlation is a mean-preserving spread]
\label{lem:cx}
Fix $p \in (0,1)$ and let $0 \le \rho < \rho' < 1$. Then
$q_\rho(W) \le_{\mathrm{cx}} q_{\rho'}(W)$ (convex order): for every convex
$f:[0,1]\to\mathbb R$,
$\mathbb E f(q_\rho(W)) \le \mathbb E f(q_{\rho'}(W))$, with strict inequality
whenever $f$ is strictly convex.
\end{lemma}

\begin{proof}
The quantile function of $q_\rho(W)$ is, using
$-\Phi^{-1}(\alpha) = \Phi^{-1}(1-\alpha)$ and monotonicity of \eqref{eq:qw}
in $w$,
\[
Q_\rho(\alpha) \;=\; \Phi\!\left( a_\rho + b_\rho\, \Phi^{-1}(\alpha)\right),
\qquad a_\rho = \frac{c}{\sqrt{1-\rho}}, \quad
b_\rho = \sqrt{\frac{\rho}{1-\rho}} .
\]
$b_\rho$ is strictly increasing in $\rho$. For $\rho<\rho'$ the affine
functions $\alpha \mapsto a_\rho + b_\rho u$ and
$a_{\rho'} + b_{\rho'} u$ of $u=\Phi^{-1}(\alpha)$ cross exactly once (their
difference is affine in $u$ with nonzero slope), hence $Q_\rho$ and
$Q_{\rho'}$ cross exactly once, with $Q_{\rho'}$ steeper: smaller on the left
of the crossing, larger on the right. Equivalently the CDFs cross exactly
once, with the $\rho'$-CDF above on the left. Both variables have mean $p$
because $\mathbb E q_\rho(W) = \Pr(Z_1 \le c) = \Phi(c) = p$ for every
$\rho$. Equal means plus single crossing of CDFs is the Karlin--Novikoff cut
criterion, which yields the convex order. Strictness for strictly convex $f$
follows from the different variances
$\Phi_2(c,c;\rho)-p^2 < \Phi_2(c,c;\rho')-p^2$, by strict
monotonicity of $\Phi_2$ in $\rho$, e.g.\ via Plackett's identity
$\partial \Phi_2(c,c;\rho)/\partial \rho = \varphi_2(c,c;\rho) > 0$.
\end{proof}

\begin{proposition}[Correlation weakens state-specific protection]
\label{prop:jensen}
For all $n \ge 2$ and $p \in (0,1)$:
\begin{enumerate}
\item[(a)] \emph{Hierarchy.} $A_{n,n}(p,\rho) = \mathbb E[q_\rho(W)^n]$ is
strictly increasing in $\rho$, and $A_{n,n}(p,\rho) \ge p^n$ with equality iff
$\rho = 0$. Applied at $(p_B,\rho_B)$: the hierarchy's false-acceptance rate
rises with $\rho_B$; the multiplicative filtering $p_B^n$ promised by
independence collapses.
\item[(b)] \emph{Polyarchy.} $A_{n,1}(p,\rho) = 1 - \mathbb E[(1-q_\rho(W))^n]$
is strictly decreasing in $\rho$, and $A_{n,1}(p,\rho) \le 1-(1-p)^n$ with
equality iff $\rho=0$. Applied at $(p_G,\rho_G)$: the polyarchy's rescue of
good proposals degrades with $\rho_G$.
\item[(c)] \emph{Majority.} For odd $n$ and $m = (n{+}1)/2$, the map
$x \mapsto \Pr(\mathrm{Bin}(n,x) \ge m)$ is strictly convex on $(0,\tfrac12)$
and strictly concave on $(\tfrac12,1)$, with inflection exactly at
$\tfrac12$. Consequently, the right derivative at $\rho=0$ given by
Proposition~\ref{prop:local} makes majority acceptance fall in state $G$
(where $p_G>\tfrac12$) and rise in state $B$ (where $p_B<\tfrac12$).
Correlation therefore increases both loss components \emph{locally at
independence}. Global monotonicity in $\rho$ is evaluated numerically on the
prespecified finite grid in C6 (\S\ref{sm:cert}).
\end{enumerate}
\end{proposition}

\begin{proof}
(a) $x^n$ is strictly convex on $[0,1]$; apply Lemma~\ref{lem:cx} for
monotonicity and Jensen's inequality at mean $p$ for the bound. Strict
monotonicity from strictness in Lemma~\ref{lem:cx}. An alternative proof of
monotonicity is Slepian's inequality: $A_{n,n} = \Pr(\max_i Z_i \le c)$ is a
Gaussian orthant probability, nondecreasing in every off-diagonal correlation.
(b) Symmetric, with $f(x) = (1-x)^n$ strictly convex, or Slepian on
$\Pr(\min_i Z_i > c)$.
(c) With $B_{n,m}(x) = \Pr(\mathrm{Bin}(n,x)\ge m)$ we have the classical
identity $B_{n,m}'(x) = n\binom{n-1}{m-1} x^{m-1}(1-x)^{n-m}$, so
$B_{n,m}''(x) = n\binom{n-1}{m-1} x^{m-2}(1-x)^{n-m-1}\,[(m-1) - (n-1)x]$.
For $m = (n{+}1)/2$ the bracket equals $(n-1)(\tfrac12 - x)$: positive for
$x<\tfrac12$, negative for $x>\tfrac12$.
\end{proof}

\begin{proposition}[Local effect of correlation]
\label{prop:local}
For any $1\le k\le n$ and $p\in(0,1)$,
\begin{equation}
\frac{\partial A_{n,k}(p,\rho)}{\partial \rho}\Big|_{\rho=0^+}
\;=\; \tfrac12\, \varphi(c)^2\, B_{n,k}''(p),
\qquad c = \Phi^{-1}(p).
\end{equation}
\end{proposition}

\begin{proof}
Write $A_{n,k}(p,\rho) = \mathbb E[B_{n,k}(q_\rho(W))]$. For every fixed
$r\ge1$, the mean-value theorem for $\Phi$ and
\[
\left|\frac{c-\sqrt\rho W}{\sqrt{1-\rho}}-c\right|
\le C\sqrt\rho\,(1+|W|)
\]
for sufficiently small $\rho$ give
$\lVert q_\rho(W)-p\rVert_r=O(\sqrt\rho)$, since Gaussian moments are finite.
Because $B_{n,k}$ is a polynomial with bounded third derivative on $[0,1]$,
Taylor's theorem, together with $\mathbb E q_\rho(W)=p$, yields
\[
A_{n,k}(p,\rho)=B_{n,k}(p)
+\tfrac12B_{n,k}''(p)\operatorname{Var}(q_\rho(W))+O(\rho^{3/2}).
\]
The variance is exactly
$\operatorname{Var}(q_\rho(W))=\Phi_2(c,c;\rho)-p^2$. By Plackett's
identity, its right derivative at zero is
$\varphi_2(c,c;0)=\varphi(c)^2$. Dividing the expansion by $\rho$ and taking
$\rho\downarrow0$ proves the result.
\end{proof}

\begin{theorem}[Collapse of fixed threshold rules at perfect correlation]
\label{thm:collapse}
For every $n$ and every $1 \le k \le n$, the acceptance probability satisfies
$A_{n,k}(p,\rho) \to p$ as $\rho \to 1^-$. Consequently, if
$(\rho_G,\rho_B)\to(1,1)$, then $R_{n,k}\to R_{1,1}$. Every fixed-size
\kofn{} threshold rule converges to the vote-following single-agent rule. The
optimum over unrestricted vote-aggregation rules is treated in
Theorem~\ref{thm:gain}.
\end{theorem}

\begin{proof}
As $\rho \to 1$, $q_\rho(W) \to \mathbf 1\{W \le c\}$ almost surely, so
$q_\rho(W) \Rightarrow \mathrm{Bernoulli}(p)$. Since
$B_{n,k}$ is continuous and bounded with $B_{n,k}(0)=0$, $B_{n,k}(1)=1$,
$\mathbb E B_{n,k}(q_\rho(W)) \to (1-p)\cdot 0 + p \cdot 1 = p$.
\end{proof}

\section{Large organizations: the correlation floor}
\label{sm:floor}

\begin{theorem}[Vasicek limit for organizations]
\label{thm:vasicek}
Fix $p\in(0,1)$, $\rho\in(0,1)$, and a fractional threshold
$\tau \in (0,1)$. Let $k_n\in\{1,\ldots,n\}$ be any integer sequence such
that $k_n/n\to\tau$. Then
\begin{equation}
\label{eq:vasicek}
\lim_{n\to\infty} A_{n,k_n}(p,\rho)
\;=\; \Phi\!\left( \frac{\Phi^{-1}(p) - \sqrt{1-\rho}\; \Phi^{-1}(\tau)}
{\sqrt{\rho}} \right).
\end{equation}
\end{theorem}

\begin{proof}
By Lemma~\ref{lem:order},
$A_{n,k_n} = \Pr(\sqrt\rho\, W + \sqrt{1-\rho}\, \varepsilon_{(k_n)} \le c)$.
By consistency of empirical quantiles,
$\varepsilon_{(k_n)} \to \Phi^{-1}(\tau)$ almost surely.
Slutsky's theorem gives
$\sqrt\rho\, W + \sqrt{1-\rho}\,\varepsilon_{(k_n)} \Rightarrow
\sqrt\rho\, W + \sqrt{1-\rho}\,\Phi^{-1}(\tau)$, whose CDF at $c$ is
\eqref{eq:vasicek} (the limit variable is continuous, so convergence of
probabilities holds at every point).
\end{proof}

\begin{corollary}[Irreducible error floor of infinite majority]
\label{cor:floor}
For strict majority, $k_n=\lceil(n{+}1)/2\rceil$ and
$k_n/n\to\tfrac12$, with $\Phi^{-1}(\tfrac12)=0$. Assume the agent is
informative in the relevant state
($p_G > \tfrac12$, i.e.\ $c_G > 0$; symmetrically $p_B < \tfrac12$, $c_B <
0$). The infinite-committee error rates are
\begin{equation}
e^\infty_G \;=\; \Phi\!\left(-\frac{c_G}{\sqrt{\rho_G}}\right)
\;>\;0,
\qquad
e^\infty_B \;=\; \Phi\!\left(\frac{c_B}{\sqrt{\rho_B}}\right)
\;>\;0,
\end{equation}
against single-agent errors $e^1_G = \Phi(-c_G)$,
$e^1_B = \Phi(c_B)$. The \emph{irreducible fraction} of
single-agent error that survives unlimited aggregation is
\begin{equation}
\label{eq:floorratio}
\mathcal F(p,\rho) \;=\;
\frac{\Phi\!\left(-c/\sqrt{\rho}\right)}{\Phi(-c)} \;\in\; (0,1)
\qquad (p > \tfrac12,\ c = \Phi^{-1}(p) > 0,\ 0<\rho<1),
\end{equation}
strictly increasing in $\rho$, with $\mathcal F \to 0$ as $\rho \to 0$ and
$\mathcal F \to 1$ as $\rho \to 1$. For the bad state, apply the formula after
the label transformation $p \mapsto 1-p$. We interpret $\mathcal F$ as an
efficiency ratio when state-specific accuracy exceeds chance; at or below
chance, the raw majority error provides the relevant quantity.
\end{corollary}

\begin{proof}
Equation~\eqref{eq:vasicek} with $\tau=\tfrac12$ gives the two error
probabilities. For $c>0$ and $0<\rho<1$,
$c/\sqrt{\rho}>c$, so
$0<\Phi(-c/\sqrt{\rho})<\Phi(-c)$ and
$\mathcal F\in(0,1)$. Because $c/\sqrt{\rho}$ decreases strictly with
$\rho$, $\mathcal F$ increases strictly. Finally,
$c/\sqrt{\rho}\to\infty$ as $\rho\downarrow0$ and
$c/\sqrt{\rho}\to c$ as $\rho\uparrow1$, which gives the two limits.
\end{proof}

Under independence ($\rho=0$) the Condorcet jury theorem drives committee
error to zero exponentially fast; \eqref{eq:floorratio} is the quantitative
correction for correlated agents (convergence to the floor illustrated in
Fig.~S2). The limiting behavior has three forms.
(i) \emph{Majority floor.} For a reference agent of accuracy $p = 0.8$
evaluated at the 5th--95th percentile of the latent correlations we measure
for open-weight cross-family large-language-model (LLM) committees (identifiable pairs, $\hat\rho \approx
\rhoCrossMin$--$\rhoCrossMax$; Section~\ref{sm:results}),
\eqref{eq:floorratio} retains between \floorFracMinPct\% and
\floorFracMaxPct\% of the single-agent error in the infinite-majority
limit. (ii) \emph{Threshold-family floor.} Varying the threshold fraction
$\tau\in[0,1]$ trades the two state errors and leaves a positive floor on
the $\kappa$-weighted loss. For
$\tau\in(0,1)$,
$R^\infty(\tau) =
[1 - \Phi(\frac{c_G - \sqrt{1-\rho_G}\,\Phi^{-1}(\tau)}{\sqrt{\rho_G}})] +
\kappa\, \Phi(\frac{c_B - \sqrt{1-\rho_B}\,\Phi^{-1}(\tau)}{\sqrt{\rho_B}})$.
Its continuous endpoint values are $R^\infty(0)=\kappa$ (accept all) and
$R^\infty(1)=1$ (reject all), and
$\min_{\tau\in[0,1]}R^\infty(\tau)>0$ when
$0<p_B<p_G<1$, $0<\kappa<\infty$, and $\rho_G,\rho_B>0$.
For every $\tau\in(0,1)$, both Gaussian tail terms are strictly positive,
and the endpoint values are also positive. Continuity on the compact interval
$[0,1]$ therefore gives a strictly positive minimum. At the
per-domain open-weight parameters we measure ($\kappa = 1$), the marginals
are averaged over open-weight models identifiable in both states and the
correlations over cross-family pairs among those same models
(\texttt{e7\_bootstrap.py}). The reported threshold-family floor is the
minimum over the finite $\tau$ grid implemented in
\texttt{theory\_lib.loss\_floor}, and is therefore an upper approximation to
the continuous infimum. On that grid it retains
\lossFloorMinPct\%--\lossFloorMaxPct\% of the mean single-agent loss across
domains. The displayed value is the finite-grid upper approximation to the
continuous best-threshold floor. (iii) \emph{Bayes floor.}
Proposition~\ref{prop:bayes} extends the statement beyond threshold
rules to \emph{all} vote-aggregation rules. Quantitative floor estimates in
this section use majority aggregation.

\begin{proposition}[Fixed-$k$ rules degenerate]
\label{prop:degenerate}
For an integer $k\geq1$ held fixed as $n$ grows, $A_{n,k} \to 1$ for every
$p \in (0,1)$, $\rho < 1$; and for the unanimity rule,
$A_{n,n} \to 0$. Thus an unboundedly grown polyarchy eventually accepts
everything and an unboundedly grown hierarchy rejects everything, in both
states; the meaningful large-$n$ design variable is the threshold fraction
$\tau = k/n$.
\end{proposition}

\begin{proof}
$\varepsilon_{(k)} \to -\infty$ a.s.\ for fixed $k$ and
$\varepsilon_{(n)} \to +\infty$ a.s.; apply Lemma~\ref{lem:order}.
\end{proof}

\begin{proposition}[Bayes floor for unrestricted rules]
\label{prop:bayes}
Assume $0<p_B<p_G<1$, $0<\kappa<\infty$, and
$0<\rho_G,\rho_B<1$. Define the Bayes risk over \emph{all}
vote-aggregation rules by
\[
R^{*}_n = \min_{\psi:\{0,1\}^n \to \{0,1\}}
\bigl[\Pr(\psi = 0 \mid G) + \kappa \Pr(\psi = 1 \mid B)\bigr].
\]
Then $R^*_n$ is nonincreasing in $n$ and
\begin{equation}
\lim_{n\to\infty} R^*_n \;=\; R^*_\infty \;=\;
\int_0^1 \min\bigl( g_G(x),\ \kappa\, g_B(x) \bigr)\, dx \;>\; 0,
\end{equation}
where $g_\theta$ is the density of $q_\theta(W)$ on $(0,1)$ (a Vasicek
density with parameters $(p_\theta, \rho_\theta)$). For
$\rho_G,\rho_B\in(0,1)$, both densities have full support, which makes the
integral strictly positive. At $\rho_\theta=0$, $q_\theta(W)$ is a point
mass in that state and the corresponding Bayes experiment uses a
mixed/discrete dominating measure.
\end{proposition}

\begin{proof}
Write $S_n=\sum_{i=1}^n V_i$.
Use the normalized prior weights
$\Pr(G)=1/(1+\kappa)$ and $\Pr(B)=\kappa/(1+\kappa)$ with unit
misclassification loss. Its ordinary Bayes risk is $R_n^*/(1+\kappa)$.
By exchangeability and sufficiency, the optimal rule depends on votes only
through $S_n$; adding an agent refines the information $\sigma$-field, so
$R^*_n$ is nonincreasing and bounded below, hence convergent. Under state
$\theta$, $S_n/n \to q_\theta(W)$ a.s.\ (conditional LLN), and $q_\theta(W)$ has
the (absolutely continuous, full-support) Vasicek density on $(0,1)$
\[
g_\theta(x) = \sqrt{\tfrac{1-\rho_\theta}{\rho_\theta}}\,
\exp\!\left( \tfrac12 \Phi^{-1}(x)^2 - \tfrac{1}{2\rho_\theta}\bigl( c_\theta
- \sqrt{1-\rho_\theta}\, \Phi^{-1}(x) \bigr)^2 \right).
\]
The limit experiment observes $X = q_\theta(W)$; its Bayes risk is obtained
by minimizing pointwise over the accept/reject decision at each $x$,
$R^*_\infty = \int_0^1 \min(g_G(x), \kappa\, g_B(x))\,dx$. To identify this
limit experiment with the information in the entire vote sequence, note that,
conditional on $(\theta,X=x)$, the votes are iid Bernoulli$(x)$: indeed
$X=q_\theta(W)$ and conditioning on its value fixes their conditional success
probability. This conditional product law is determined by $x$ and is
invariant across $\theta$, so
$\theta\perp\!\!\!\perp \sigma(V_1,V_2,\ldots)\mid X$. Conversely,
$X=\lim_{n\to\infty}S_n/n$ almost surely, so $X$ is measurable with respect
to the infinite-vote $\sigma$-field. Hence the posterior given $V_{1:n}$ is a
bounded martingale that converges almost surely to the posterior given $X$.
Applying dominated convergence to the bounded ordinary Bayes loss and
multiplying by $1+\kappa$ gives
$R^*_n \downarrow R^*_\infty$. Overlap of full-support densities forces
$R^*_\infty>0$.
\end{proof}

\begin{remark}[Unanimity can signal a correlated-error state]
\label{rem:toogood}
As a function of $u = \Phi^{-1}(x)$, the log-likelihood ratio
$\log g_G(x)/g_B(x)$ is \emph{quadratic}, with leading coefficient
$\tfrac12\left(\tfrac{1-\rho_B}{\rho_B} - \tfrac{1-\rho_G}{\rho_G}\right)$.
It is monotone in $x$ iff $\rho_G = \rho_B$. When $\rho_B > \rho_G$ the
log-likelihood ratio is strictly concave in $u$. If its maximum reaches
$\log\kappa$, the Bayes acceptance set is a bounded (possibly degenerate)
interval of vote shares; a maximum below $\log\kappa$ yields rejection at
every vote share.
Thus, whenever the acceptance set is nonempty, both extremes are rejected:
an overwhelming approval consensus can signal the correlated-error state and
provide weaker evidence of quality, a phenomenon also studied in forensic
identification.
\end{remark}

\section{Optimal threshold within the $k$-of-$n$ family}
\label{sm:optimal}

\paragraph{Equal-correlation monotone likelihood ratio.}
If $0<p_B<p_G<1$ and
$\rho_G=\rho_B=\rho\in[0,1)$, the vote-count likelihood ratio
$\ell(s)=\Pr(S=s\mid G)/\Pr(S=s\mid B)$ is strictly increasing in $s$.
For $\rho=0$ this follows directly from the binomial likelihood ratio.
For $0<\rho<1$, let $g_\theta$ be the density of
$X=q_\theta(W)$ and set $u=\Phi^{-1}(x)$. The density in
Proposition~\ref{prop:bayes} gives
\[
\log\frac{g_G(x)}{g_B(x)}
=\frac{c_B^2-c_G^2}{2\rho}
+ \frac{\sqrt{1-\rho}\,(c_G-c_B)}{\rho}\,u ,
\]
which is strictly increasing in $x$ because $c_G>c_B$. Write this density
ratio as $h(x)$ and
$K_s(x)=\binom ns x^s(1-x)^{n-s}$. If
$\mu_s(dx)\propto K_s(x)g_B(x)\,dx$, then
$\ell(s)=\int h\,d\mu_s$. Moreover,
$d\mu_{s+1}/d\mu_s$ is proportional to $x/(1-x)$, so
$\mu_{s+1}$ strictly likelihood-ratio dominates $\mu_s$. The expectation of
the strictly increasing function $h$ therefore increases strictly with $s$.

\begin{theorem}[Optimal threshold]
\label{thm:kstar}
Assume $0<p_B<p_G<1$, $0<\kappa<\infty$, and
$0\leq\rho_G,\rho_B<1$.
Define $\Delta(k) = \Pr(S = k \mid G) - \kappa \Pr(S = k \mid B)$ for
$k = 1,\dots,n-1$, so that $R_{n,k+1} - R_{n,k} = \Delta(k)$ exactly. The
optimal threshold set is $\arg\min_k R_{n,k}$; if the vote-count likelihood
ratio $\ell(s) = \Pr(S=s \mid G)/\Pr(S=s\mid B)$ is nondecreasing (MLR), then
\begin{equation}
\begin{aligned}
k_-^*
&=\min\Bigl(
\{k\in\{1,\ldots,n-1\}:\ell(k)\ge\kappa\}\cup\{n\}
\Bigr),\\
k_+^*
&=\min\Bigl(
\{k\in\{1,\ldots,n-1\}:\ell(k)>\kappa\}\cup\{n\}
\Bigr),
\end{aligned}
\end{equation}
and
\[
\arg\min_{1\le k\le n}R_{n,k}
=\{k_-^*,k_-^*+1,\ldots,k_+^*\}.
\]
Thus $k_-^*$ is the smallest optimizer and $k_+^*$ the largest; when
$\ell(k)\ne\kappa$ for every $k=1,\ldots,n-1$, the optimizer is unique.
\end{theorem}

\begin{proof}
Since $A_{n,k} - A_{n,k+1} = \Pr(S = k)$,
\begin{align*}
R_{n,k+1} - R_{n,k}
&= \bigl[A_{n,k}(p_G,\rho_G) - A_{n,k+1}(p_G,\rho_G)\bigr]
- \kappa\bigl[A_{n,k}(p_B,\rho_B) - A_{n,k+1}(p_B,\rho_B)\bigr] \\
&= \Pr(S_G = k) - \kappa \Pr(S_B = k) = \Delta(k).
\end{align*}
Raising the threshold from $k$ to $k{+}1$ flips the decision on the
borderline event $\{S = k\}$ from accept to reject, which pays iff those
cases are likelier to be bad: $\Delta(k) \le 0 \iff \ell(k) \le \kappa$.
Under MLR, $R_{n,k}$ decreases strictly for $\ell(k)<\kappa$, is flat
across differences for which $\ell(k)=\kappa$, and increases strictly once
$\ell(k)>\kappa$. This gives the stated interval of optimizers.
\end{proof}

Unequal state correlations can produce non-MLR vote counts (check C8b),
consistent with the analytic vote-share limit in
Remark~\ref{rem:toogood}. All empirical analyses therefore compute $k^*$
by direct enumeration of \eqref{eq:loss}, a procedure valid for general
vote-count distributions.

\begin{proposition}[Monotone comparative statics]
\label{prop:topkis}
Let $S_G$ and $S_B$ have any fixed state-conditional vote-count
distributions for a committee of size $n$, and write
$R_{n,k}(\kappa)=1-\Pr(S_G\ge k)+\kappa\Pr(S_B\ge k)$.
(a) The argmin correspondence
\[
\kappa \longmapsto \arg\min_k R_{n,k}(\kappa)
\]
is nondecreasing in the strong set order; in particular both the smallest and
the largest optimal threshold are nondecreasing in $\kappa$:
costlier false acceptances demand stricter organizations (toward
hierarchy); cheaper ones demand laxer organizations (toward polyarchy).
The proposition applies to arbitrary fixed state-conditional vote-count
distributions. Our implementation reports the smallest optimizer, which is
nondecreasing under the result above.
(b) If $0<p_B<p_G<1$ and $0<\kappa<\infty$, then at
$\rho_G = \rho_B = 0$ the rule reduces to the classical
Sah--Stiglitz/Nitzan--Paroush threshold. With the smallest-optimizer tie
convention used here,
\[
k_-^*=\min\left(
\left\{k\in\{1,\ldots,n-1\}:
k\log\frac{p_G}{p_B}
+(n-k)\log\frac{1-p_G}{1-p_B}\ge\log\kappa
\right\}\cup\{n\}
\right).
\]
\end{proposition}

\begin{proof}
(a) The first difference $\Delta(k;\kappa) = R_{n,k+1}-R_{n,k} =
\Pr(S_G=k) - \kappa\,\Pr(S_B=k)$ is \emph{nonincreasing} in $\kappa$ for
every $k$: the loss function has decreasing differences in $(k, \kappa)$.
By Topkis's theorem on monotone comparative statics, the argmin
correspondence $\kappa \mapsto \arg\min_k R_{n,k}(\kappa)$ is nondecreasing
in the strong set order, so the smallest and largest selections
$k^*(\kappa)$ are nondecreasing.
(Directly, let $D(\kappa)=R_{n,k}(\kappa)-R_{n,k'}(\kappa)
=\sum_{s=k'}^{k-1}\Delta(s;\kappa)$ for $k'<k$. If $k$ is optimal at
$\kappa$ and $k'$ is optimal at $\kappa'>\kappa$, then
$D(\kappa)\le0\le D(\kappa')$. Moreover,
\[
D(\kappa')=D(\kappa)
-(\kappa'-\kappa)\sum_{s=k'}^{k-1}\Pr(S_B=s)
\le D(\kappa).
\]
Hence $0\le D(\kappa')\le D(\kappa)\le0$; both differences vanish and so
does $\Pr(k'\le S_B<k)$. Thus both thresholds have identical loss at both
values of $\kappa$, and $k$ is also optimal at $\kappa'$.)
(b) Substitute the binomial pmf into $\ell$:
$\ell(s) = (p_G/p_B)^s \bigl((1-p_G)/(1-p_B)\bigr)^{n-s}$, increasing in
$s$; apply Theorem~\ref{thm:kstar}.
\end{proof}

\begin{proposition}[Inputs to fixed-committee threshold design]
\label{prop:knobs}
For a homogeneous exchangeable committee of fixed size $n$ and a specified
loss ratio $\kappa$, the loss of every $k$-of-$n$ rule is computable from the
four state-conditional statistical parameters
$(p_G,p_B,\rho_G,\rho_B)$ via \eqref{eq:Ank}. Together with $\kappa$, these
parameters determine the optimal threshold set
$\arg\min_{1\le k\le n}R_{n,k}$.

If $n$ is itself a design choice, an admissible size set or an explicit
member-cost or budget model is additionally required. In the heterogeneous
extension of Section~\ref{sm:hetero}, the corresponding inputs are the
marginal vectors $(p_{G,i})_{i=1}^n$, $(p_{B,i})_{i=1}^n$ and the two
state-specific correlation matrices $\Sigma_G,\Sigma_B$. Section~\ref{sm:est}
gives consistency of plug-in threshold selection for fixed $n$.
\end{proposition}

\begin{proof}
Lemma~\ref{lem:rep} and \eqref{eq:loss} make every $R_{n,k}$ a
deterministic function of the listed inputs. Minimization over the finite set
$\{1,\ldots,n\}$ gives the optimal threshold set. Allowing committee size or
heterogeneous members changes the design space and therefore requires the
additional inputs stated in the proposition.
\end{proof}

\section{Limits of fixed-size committee gains}
\label{sm:paradox}

The fixed-committee correlation limit extends from threshold rules to
arbitrary vote-only aggregation.

\begin{theorem}[Correlation limits fixed-size organizational gain]
\label{thm:gain}
Fix $0<p_B<p_G<1$ and $0<\kappa<\infty$.
Define the organizational gain at parameters
$(p_G,p_B,\rho_G,\rho_B;\kappa)$ as
$\Gamma_n = R_{1,1} - \min_k R_{n,k} \ \ge 0$.
\begin{enumerate}
\item[(a)] At $\rho_G=\rho_B=\rho$, $\Gamma_n\to0$ as $\rho\to1$ for every
\emph{fixed} $n$ (Theorem~\ref{thm:collapse}).
\item[(b)] More generally, for every fixed finite size bound $N$,
\[
\lim_{(\rho_G,\rho_B)\to(1,1)}
\max_{1\le n\le N}\Gamma_n=0.
\]
Let $R_n^*$ be the Bayes risk over all aggregation rules, as in
Proposition~\ref{prop:bayes}, and let
\[
R_1^*=\min(p_G,\kappa p_B)+\min(1-p_G,\kappa(1-p_B))
\]
be the Bayes-optimal one-vote risk. The same bounded-size collapse holds for
unrestricted rules:
\[
\lim_{(\rho_G,\rho_B)\to(1,1)}
\max_{1\le n\le N}|R_n^*-R_1^*|=0.
\]
\end{enumerate}
\end{theorem}

\begin{proof}
First, $\Gamma_n\ge0$: for every realized vote count $S=s$, exactly $s$ of
the $n$ thresholds accept, and therefore
$\frac1n\sum_{k=1}^n A_{n,k}(p,\rho)=\mathbb E[S/n]=p$. Consequently,
$\frac1n\sum_{k=1}^n R_{n,k}=R_{1,1}$ and
$\min_kR_{n,k}\le R_{1,1}$.
Part (a) is Theorem~\ref{thm:collapse}. For part (b), for each fixed pair
$(n,k)$, Theorem~\ref{thm:collapse} applied state by state gives
$R_{n,k}\to R_{1,1}$ as $(\rho_G,\rho_B)\to(1,1)$. The collection
$\{(n,k):1\le k\le n\le N\}$ is finite, so the convergence is uniform over
that collection, proving the first display.

For unrestricted rules and fixed $n$, under state $\theta$ the vote-vector
distribution converges in total variation to the two-point distribution
placing mass $p_\theta$ on $(1,\ldots,1)$ and mass $1-p_\theta$ on
$(0,\ldots,0)$; the probability of every non-unanimous vector vanishes.
Bayes risk on the finite vote space is a continuous function of its two
state-conditional probability masses, so $R_n^*\to R_1^*$. Finiteness of
$N$ again makes the convergence uniform over $1\le n\le N$.
\end{proof}

\paragraph{Finite-grid observation.}
For majority committees in the operating region
($p_B<\tfrac12<p_G$), Proposition~\ref{prop:jensen}c proves that common
correlation increases loss locally at $\rho=0^+$. Majority loss was also
nondecreasing along $\rho_G=\rho_B=\rho$ in all 2{,}700 parameter curves of
check C6 (\S\ref{sm:cert}), supplying numerical evidence over the
prespecified operating region. The single-agent risk $R_{1,1}$ is constant in
$\rho$, so majority gain was correspondingly nonincreasing on that grid.

Theorem~\ref{thm:gain} takes the near-perfect-correlation limit uniformly over
a fixed finite range of committee sizes. Proposition~\ref{prop:bayes} fixes
$(\rho_G,\rho_B)$ below one and sends $n\to\infty$, revealing the continuous
statistic $q_\theta(W)$. The large-committee limit can retain information
absent from one binary vote, so the two limit orders can produce different
results.

\section{Heterogeneous agents and estimation}
\label{sm:est}
\label{sm:hetero}

\subsection{Heterogeneous model}
Agents may differ in marginals and correlations: state-$\theta$ latent scores
$Z^{(\theta)} \sim \mathcal N(0, \Sigma_\theta)$ with unit diagonal, and
$V_i = \mathbf 1\{Z^{(\theta)}_i \le \Phi^{-1}(p_{\theta,i})\}$. The
acceptance probability of any $k$-of-$n$ rule is the probability that at
least $k$ coordinates of a correlated Gaussian vector fall below their
thresholds. Estimated correlation matrices are projected to the
positive-semidefinite cone by eigenvalue clipping followed by rescaling to
unit diagonal.

\subsection{Estimation protocol}
From $N$ items of known state we estimate, per state $\theta$:
(i) marginals $\hat p_{\theta,i}$ (vote frequencies);
(ii) pairwise latent correlations $\hat\Sigma_{\theta,ij}$ by tetrachoric
inversion of the $2\times2$ vote table after constraining its joint
probability to the binary Fr\'echet bounds;
(iii) a common-factor summary $\hat\rho_\theta$ = mean off-diagonal latent
correlation;
(iv) the chance-adjusted agreement (binary CAPA analogue,
$(\mathrm{obs} - \mathrm{exp})/(1 - \mathrm{exp})$) for comparability with
the model-similarity literature.
For the exchangeable one-factor comparator, negative mean off-diagonal
estimates are projected to zero because its parameter space requires
$\rho_\theta\geq0$. The heterogeneous full-matrix model retains signed
pairwise estimates subject to positive-semidefinite regularization.
Experimental resampling, Monte Carlo evaluation, and held-out threshold
selection are specified in \S\ref{sm:methods}.

\begin{proposition}[Exact plug-in selection is consistent]
\label{prop:consistent}
Assume state-conditional items are independent draws from the specified
heterogeneous Gaussian-copula model, with the number of observations in each
state tending to infinity. Fix the design space
$\mathcal K=\{1,\ldots,n\}$ and suppose the true loss
vector $(R_{n,k})_{k\in\mathcal K}$ has a unique minimizer $k^*$. Let
$\widetilde R_{N,n,k}$ be the \emph{exact} Gaussian-copula loss evaluated at
consistent marginal and correlation estimates from $N$ items. Assume the
true marginals and correlation matrices are interior points of their
parameter spaces and that any PSD regularization either becomes inactive or
uses an eigenvalue floor $\epsilon_N\downarrow0$. Then
\[
\Pr\!\left(\arg\min_{k\in\mathcal K}\widetilde R_{N,n,k}=\{k^*\}\right)
\longrightarrow1.
\]
For Monte Carlo orthant estimates based on $M_N$ draws, the same conclusion
holds when $M_N\to\infty$.
\end{proposition}

\begin{proof}
Consistency of the $2\times2$ cell frequencies follows from the LLN, and
the tetrachoric inverse is continuous at interior marginals because
$\Phi_2$ is strictly increasing in $\rho$ (Plackett's identity). Under the
stated regularization condition, the estimated matrices converge to the true
matrices. Gaussian orthant probabilities are continuous at interior
$(p,\Sigma)$, so
$\widetilde R_{N,n,k}\to R_{n,k}$ in probability for every $k$. Because
$\mathcal K$ is finite, convergence is uniform over $k$; the positive loss
gap at the unique minimizer then gives exact plug-in selection consistency.
With $M_N\to\infty$, the Monte Carlo errors also converge uniformly to zero
over the finite set $\mathcal K$, establishing selection consistency.
\end{proof}

\paragraph{Scope across homogeneous and heterogeneous committees.}
Finite-roster threshold selection and exchangeable-limit results use
different assumptions. The former is determined by the two state-conditional
vote-count distributions; the latter uses the homogeneous one-factor structure.
\begin{center}
\small
\begin{tabular}{@{}>{\raggedright\arraybackslash}p{0.25\textwidth}
>{\raggedright\arraybackslash}p{0.32\textwidth}
>{\raggedright\arraybackslash}p{0.35\textwidth}@{}}
\hline
Claim & Homogeneous exchangeable committee & Fixed heterogeneous committee \\
\hline
Loss and optimal threshold
& Mixed-binomial quadrature; enumerate $k=1,\ldots,n$
& Multivariate Gaussian probabilities; enumerate $k=1,\ldots,n$ \\
Monotonicity in $\kappa$
& Any fixed state-conditional vote-count distributions
  (Prop.~\ref{prop:topkis})
& Holds for any fixed state-conditional vote-count distributions
  (Prop.~\ref{prop:topkis}) \\
Likelihood-ratio crossing rule
& MLR is proved when $\rho_G=\rho_B$; unequal correlations can violate it
& Applies when the fitted vote-count distributions satisfy MLR; direct
  enumeration applies generally \\
Correlation effects, error floor, and near-perfect-correlation collapse
& Derived from homogeneous marginals and the exchangeable common factor
& Require additional structure, including a specified sequence of growing
  rosters for a large-$n$ limit \\
Plug-in threshold selection
& Included as a special case
& Consistent for fixed $n$ under the conditions of
  Prop.~\ref{prop:consistent} \\
\hline
\end{tabular}
\end{center}

\subsection{Relation between CAPA and the latent correlation}
For two agents with equal accuracy $a$ on a balanced binary task, both the
chance-adjusted agreement $\mathrm{CAPA}$ and the manifest error correlation
are monotone transforms of the latent $\rho$ at fixed marginals; they differ
in how they aggregate over states. State-conditional $\rho_\theta$ separates
agreement induced by high marginal accuracy from joint state-specific
errors. Concretely, for two agents with common state-$\theta$ marginal
$p_\theta$ the state-conditional agreement probability is
\[
P^{\mathrm{agree}}_\theta \;=\; \Phi_2(c_\theta, c_\theta; \rho_\theta)
\;+\; \Phi_2(-c_\theta, -c_\theta; \rho_\theta),
\qquad c_\theta = \Phi^{-1}(p_\theta),
\]
and the state-conditional chance-adjusted agreement is
$\mathrm{CAPA}_\theta = (P^{\mathrm{agree}}_\theta - e_\theta)/(1 -
e_\theta)$ with $e_\theta = p_\theta^2 + (1-p_\theta)^2$ its value at
$\rho_\theta = 0$. By Plackett's identity $\mathrm{CAPA}_\theta$ is a
strictly increasing function of $\rho_\theta$ at fixed $p_\theta$, so the
two quantities order agent pairs identically \emph{within} a state and a
marginal profile. The slope of the transform depends strongly on $p_\theta$,
so equal agreement values can encode very different latent correlations at
different approval-rate levels, and pooled (state-blind)
agreement metrics additionally mix $\rho_G$ and $\rho_B$ with weights set
by the item mix. Figure~S6 shows the empirical association; variation in
marginals across pairs can produce a nonmonotone plotted relation even within
a state. Agreement-based similarity metrics provide a preliminary screen;
threshold design uses the state-conditional pair
$(\rho_G, \rho_B)$, which must be estimated from labeled screening
decisions.

\clearpage
\section{Numerical stress tests and implementation checks}
\label{sm:cert}

Table~\ref{tab:cert} reports finite-grid checks of formula implementations,
limiting cases, and selected inequalities. Analytic results establish the
continuous-space properties. C6 evaluates majority-loss monotonicity on the
stated operating grid, and C8 validates the implementation of the proved
equal-correlation MLR result. The checking harness is
\texttt{e1\_theory\_cert.py} (in
\texttt{codes/}); it writes \texttt{e1\_theory\_cert.json} (in
\texttt{outputs/}), summarized in Table~\ref{tab:cert}.

\begin{table}[h]
\centering\footnotesize\setlength{\tabcolsep}{4pt}
\begin{tabular}{@{}c >{\raggedright\arraybackslash}p{0.40\textwidth} >{\raggedright\arraybackslash}p{0.17\textwidth} >{\raggedright\arraybackslash}p{0.26\textwidth}@{}}
\hline
Check & Quantity or property checked & Grid & Result \\
\hline
C1 & Quadrature matches Monte Carlo ($2{\times}10^6$ draws) & 40 random configs & max dev.\ $7.0\times10^{-4}$ (within MC noise) \\
C2 & Endpoints: $\rho{=}0$ binomial exact; $\rho{\to}1$ collapse & 160 configs (320 checks) & errors $0$ and $<5\times10^{-3}$ \\
C3 & $A_{n,n}\uparrow\rho$;\ $A_{n,1}\downarrow\rho$ (Prop.~\ref{prop:jensen}) & 155 parameter settings, two curves each, $\times$ 39 $\rho$ & 0 violations \\
C4 & Jensen bounds & 155 configs & 0 violations \\
C5 & $\partial A/\partial\rho|_0 = \tfrac12\varphi(c)^2 B''_{n,k}(p)$ (Prop.~\ref{prop:local}) & 6 configs & max rel.\ err.\ $6.4\times10^{-4}$ \\
C6 & Majority loss $\uparrow \rho$ on the tested operating grid & 2{,}700 parameter curves & 0 violations on grid \\
C7 & Convergence to Vasicek floor (Thm.~\ref{thm:vasicek}) & 3 trajectories to $n{=}501$ & gap $<2\times10^{-3}$ \\
C8 & MLR of vote counts when $\rho_G=\rho_B$ on the tested grid & 1{,}620 configs & 0 violations on grid \\
C8b & MLR when state correlations vary independently & 162 configs & 38 violations, all at unequal correlations \\
C9 & $k^*(\kappa)$ nondecreasing (Prop.~\ref{prop:topkis}) & 81 curves & 0 violations \\
C10 & Fixed-size collapse, including extreme $\kappa$ & 180 threshold configs; 40 Bayes endpoints & exact errors $0$; near-$\rho$ scaled err.\ $6.8\times10^{-4}$ \\
C11 & Exact likelihood-ratio ties and smallest-optimizer convention & 126 configs & 0 violations; max gap $2.3\times10^{-16}$ \\
\hline
\end{tabular}
\caption{\textbf{Finite-grid numerical stress tests.} Violations count
tested grid points at which the stated property fails beyond numerical
tolerance. Results are reported in their native units. A zero is the observed
violation count on the specified finite grid; the corresponding analytic
results are cited above. C8b records counterexamples to MLR under
unequal state correlations and illustrates the non-monotone regime in Fig.~S1.}
\label{tab:cert}
\end{table}

\begin{figure}[h]
\centering
\includegraphics[width=0.85\textwidth]{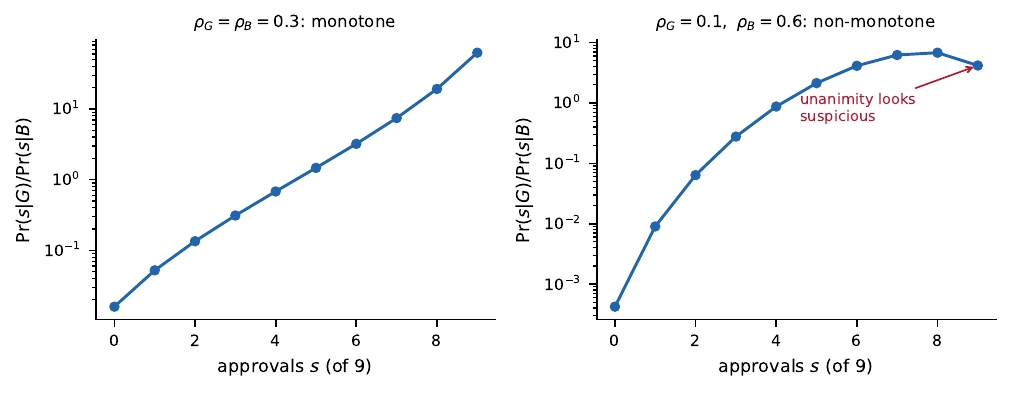}
\caption{\textbf{Unequal state correlations can make the vote-count
likelihood ratio nonmonotone.}
Vote-count likelihood ratios $\Pr(s\,|\,G)/\Pr(s\,|\,B)$ for a
9-member committee ($p_G = 0.75$, $p_B = 0.25$). Left, equal correlations,
monotone. Right, $\rho_B > \rho_G$ makes the ratio non-monotone at high
vote counts; unanimous approval carries a \emph{lower} likelihood ratio
than 8-of-9 approval and therefore gives weaker evidence of quality even
at the higher vote count.}
\end{figure}

\begin{figure}[h]
\centering
\includegraphics[width=0.5\textwidth]{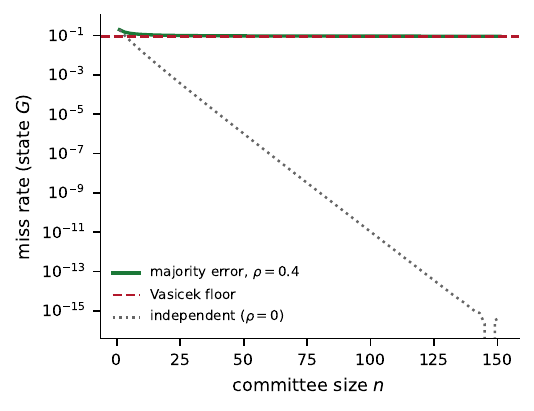}
\caption{\textbf{Convergence to the correlation floor.}
Majority-committee miss rate versus committee size at $p = 0.8$:
independent agents obey Condorcet (dotted, error $\to 0$); correlated agents
($\rho = 0.4$) converge to the Vasicek floor (dashed line;
Eq.~2 of the main text).}
\end{figure}

\clearpage
\section{Experimental methods}
\label{sm:methods}

\subsection{Screening domains}
All four domains present a proposal with hidden binary state; the agent must
answer with exactly one word, \emph{approve} or \emph{reject}. Item counts
are balanced across states except where noted.

\textbf{Fact verification (\emph{facts}).} 400 claim--evidence pairs sampled
from the VitaminC test split \cite{schuster2021vitaminc}; state $G$ =
SUPPORTS (approve), state $B$ = REFUTES (reject). The agent sees the
evidence sentence and the claim.

\textbf{Mathematical solution verification (\emph{math}).} 400 GSM8K test
problems \cite{cobbe2021gsm8k} with worked solutions. State $G$: the
dataset's reference solution (calculator annotations stripped). State $B$:
the final calculator step's result is perturbed by a plausible delta
($\pm1$--$3$, $\pm10$, halving, or comparable), the perturbed value is
propagated through the remaining text including the final answer line, so
the solution contains exactly one arithmetic inconsistency.

\textbf{Code review (\emph{code}).} 357 MBPP (sanitized) tasks
\cite{austin2021mbpp}. State $G$: the reference program, verified to pass
all official test assertions. State $B$: a single-site mutation of the
reference (comparison or arithmetic-operator swaps, off-by-one edits, boundary
index changes, boolean flips), \emph{verified to fail} at least one official
test and remain syntactically valid. Each task appears exactly once,
as either $G$ or $B$, to preserve independence across items. (Mutation
search exhausted eligible mutants for some tasks, yielding 200 $G$ / 157 $B$.)

\textbf{Truthfulness screening (\emph{truth}).} 400 TruthfulQA
multiple-choice questions \cite{lin2022truthfulqa}; state $G$ pairs the
question with its best correct answer, state $B$ with the most plausible
incorrect (misconception) answer.

\subsection{Agents}
\nModelsOpen{} open-weight instruction-tuned models served with ollama
(Q4\_K\_M quantization) on one verified NVIDIA A100 80GB GPU:
Qwen2.5 \{0.5B, 1.5B, 3B, 7B, 14B, 32B\}, Llama-3.1-8B, Llama-3.2-3B,
Mistral-7B, Mistral-Small-24B, Gemma-2 \{2B, 9B\}, Gemma-4-31B, Phi-4-14B,
GLM-4-9B, Yi \{9B, 34B\}, OLMo-2-13B, InternLM2-20B, and
Command-R-35B; plus \nModelsFrontier{} API models, one per provider
(DeepSeek V4 Pro, Zhipu GLM-5.2, Alibaba Qwen3.7-Max, Moonshot Kimi-K2, xAI
Grok-4.5, Amazon Nova Pro, Anthropic Claude Opus 4.8, OpenAI GPT-5.5) under
the identical prompt protocol with provider-side reasoning disabled, so that
all judges emit direct verdicts. DeepSeek, Zhipu, Alibaba, and Moonshot were
queried through provider APIs; the other API models were queried through the
gateway routes recorded with each vote.

For parameter-tier summaries, the open-weight models are partitioned before
analysis into three disjoint bins: small (0.5--4B: Qwen2.5-0.5B, 1.5B, 3B;
Gemma-2-2B; Llama-3.2-3B), mid (7--9B: Qwen2.5-7B, Llama-3.1-8B,
Mistral-7B, Gemma-2-9B, GLM-4-9B, Yi-9B), and top (13--35B: OLMo-2-13B,
Qwen2.5-14B, Phi-4-14B, InternLM2-20B, Mistral-Small-24B, Gemma-4-31B,
Qwen2.5-32B, Yi-34B, Command-R-35B). The labels describe parameter-count
bins used in the cross-sectional tier comparisons.

Each (model, domain, item) is voted on four times: once greedily (temperature
0) and in three temperature-0.7 samples. Local samples use fixed seeds; API
samples use provider-managed stochastic sampling. Collector-level evidence
records an 8-token default for other base local terse runs.
Gemma-4-31B, Mistral-Small-24B, Yi-34B, OLMo-2-13B, InternLM2-20B, and
Command-R-35B were launched with \texttt{THINK\_OFF=1} and
\texttt{NUM\_PREDICT=64}; an archived Phi-4 code-review refill used
\texttt{NUM\_PREDICT=8}. API collectors used a collector-path-specific
limit of 8 or 16 tokens.

The system prompt frames the agent as a reviewer on an evaluation board; the
user prompt is the domain template above. Responses are parsed from the first
occurrence of approve or reject, followed by yes or no. Responses without either
keyword pass through a uniform extended parser (earliest match among
correct/true/valid/supports/supported/agrees/yes/lgtm/accurate versus
incorrect/wrong/buggy/flawed/opposed/disagreed/false/unsupported/fails/
invalid/no, with ties resolved to reject). Records unresolved by this parser
were assigned reject in the archived primary analysis. Agents whose
remaining-unparsed rate on a domain exceeds 15\% are excluded from that
domain's pools. This criterion removes Phi-4 from code review.

\subsection{Reviewed-record parser sensitivity analysis}
\label{sm:parseraudit}

Parser sensitivity was assessed on a reviewed set containing all
\ParserAuditUnresolvedCensusN{} unresolved primary records and a deterministic
stratified sample of 2,000 parsed records, including all 187 resolved only by
the extended parser. This gives \ParserAuditReviewedTotalN{} unique reviewed
records.

Pass A used semantic language-model review, pass B applied predefined
deterministic rules, and pass C used a language model to adjudicate the
\ParserAuditAdjudicatedDisagreementsN{} disagreements. Passes A and B were
independently blinded to the parser labels, adjudication key, and each other's
output; the materials identified the audit subset. They agreed on
\ParserAuditAgreementRecordsN{} records. Pass C received both prior
assessments and remained blinded to the parser labels and adjudication key. It
assigned binary labels to 307 disagreements and left 88 indeterminate.

The primary reviewed-record overlay contains \ParserAuditPrimaryOverlayN{} rows:
\ParserAuditPrimaryOverlayApproveN{} approve,
\ParserAuditPrimaryOverlayRejectN{} reject, and
\ParserAuditPrimaryOverlayMissingN{} indeterminate. Inverse sampling-fraction
weights map these reviewed rows to the
\ParserAuditPrimaryWeightedPopulationN{}-record primary population. Weighted
binary coverage is \ParserAuditPrimaryBinaryCoveragePct{}\%, and weighted binary
label discordance with the archived parser, among records with a binary
reviewed label, is \ParserAuditPrimaryWeightedLabelDiscordancePct{}\%. These
weighted quantities are descriptive audit summaries of coverage and parser
discordance among records receiving binary reviewed labels.

Among the \ParserAuditUnresolvedCensusN{} unresolved primary records, review
yielded \ParserAuditUnresolvedApproveN{} approve,
\ParserAuditUnresolvedRejectN{} reject, and
\ParserAuditUnresolvedIndeterminateN{} indeterminate labels. The main overlay uses
listwise deletion for indeterminate reviewed records. A forced-reject point
sensitivity retains all binary reviewed labels and maps the remaining
indeterminate overlay labels to reject. An agreement-only overlay sets all
C-adjudicated primary-overlay disagreements to missing: 358 such rows were in
the primary overlay, of which 84 were already missing, so this policy adds
\ParserAuditAgreementOnlyAdditionalMissingN{} missing records for
\ParserAuditAgreementOnlyOverlayMissingN{} in total; the other 37 C-adjudicated
rows lie outside the primary overlay. These policies enter the
design-balanced selection analysis as point sensitivities.

\subsection{Estimation and validation protocol}
Vote matrices comprise primary-protocol files assembled by domain over agents,
where an agent is a model--sampling-run pair.
Items are split by parity, with \emph{odd} items forming the estimation half
(marginals $\hat p_{\theta,i}$; pairwise tetrachoric latent correlations
$\hat\Sigma_\theta$ after constraining joint probabilities to the binary
Fr\'echet bounds; PSD projection) and \emph{even}
items the evaluation half (realized architecture losses). For each committee
pool (same-model temperature samples; cross-model single-sample; full mixed
pool) we draw up to 150 random committees per size $n \in \{3,5,7,9\}$ and
evaluate all thresholds $k = 1,\dots,n$ at
$\kappa \in \{0.25, 0.5, 1, 2, 4\}$. This grid spans cost-weighted prior
odds from one quarter to four and receives equal design weight. A deployment
analysis substitutes the calibrated value or interval defined in \S S1. The floor formulas
(\S S4) and the Topkis comparative statics (\S S5) hold for every fixed,
finite, positive $\kappa$. Sampled-composition prediction calibration uses fixed-seed
Monte Carlo with $2\times10^5$ draws per fitted committee. Design-balanced
primary selection inference uses
\SelectionMonteCarloDraws{} antithetic common random draws within each refit so
that all candidate thresholds share the same numerical noise. Baselines use
the same machinery with $\Sigma$ forced to identity or the majority
threshold. Figure~3D also shows the evaluation-sample oracle, the minimum
loss across candidate thresholds in that finite evaluation sample.

\textbf{Sampled-composition descriptive calibration.}
This analysis evaluates \robustAblationN{} sampled threshold--cost rows. The pooled
identity-line statistic fixes the slope at one and the intercept at zero:
$R^2_{\mathrm{id}}=1-\sum_j(y_j-\hat y_j)^2/\sum_j(y_j-\bar y)^2$.
It spans four domains and five values of
$\kappa$, whose scaling
and unequal row counts contribute to the total sum of squares; within domains
$R^2 = \calibRsqDomMin$--$\calibRsqDomMax$ (Fig.~\ref{fig:smcal}), and
within the twenty domain $\times$ cost-ratio cells
$R^2 = \calibRsqCellMin$--$\calibRsqCellMax$ (median
\calibRsqCellMedian). Its committee-cluster bootstrap
($B=\robustBootB$) retains all rows of a committee together and holds
evaluation-half items and fitted parameters fixed. The resulting intervals
describe the sampled composition.
The exchangeable ablation is the constrained nonnegative one-factor
projection. For each committee and state, a negative mean off-diagonal
tetrachoric estimate is set to zero before evaluating the one-factor model.
The heterogeneous full-matrix model prediction retains admissible
negative pairwise estimates after
positive-semidefinite projection.

\textbf{Design-balanced selection inference.}
The primary estimand defines 20 domain $\times$ pool-kind $\times$
committee-size cells:
self-$n{=}3$, cross-$n{=}3,5,7$, and mixed-$n{=}9$ in each domain. Every cell
has target weight 0.05, and the five values of $\kappa$ for a committee remain one
cluster. In each of \SelectionBootstrapValid{} valid replicates, estimation and
evaluation items are sampled independently with replacement within domain
and state; all agent marginals and pairwise tetrachoric correlations are
re-estimated; each committee submatrix is projected to the
positive-semidefinite cone; all thresholds are re-predicted using
\SelectionMonteCarloDraws{} antithetic common draws; and the smallest numerical
minimizer is reselected separately for the full-matrix dependence model and
the independence model.
Committee clusters are then resampled within their design cells. Shared item
draws and globally re-estimated pair statistics preserve dependence among
committees sharing agents.

The pooled design-cell-balanced estimand is the prespecified primary estimand
and has three marginal contrast intervals.
Sampled-composition, pool-kind-balanced, self-only, cross/mixed-only,
within-domain, model-block, and family/model-block summaries are exploratory
sensitivities with pointwise intervals. All inference is
conditional on the observed model roster, four benchmark domains, and stated
committee design. Parameter summaries reported
elsewhere use their stated stratified item bootstraps ($B=1{,}000$).

\textbf{Pair universes.} All pairwise latent correlations use the same
tetrachoric estimator and identifiability window, with marginals in
$(0.05,0.95)$ and at least eight items per state. The pooled open-weight
cross-family summaries in Figs.~1C and 2B require both states to satisfy
these conditions for both members. This yields \nCrossPairs{} pairs across
the four domains. The family analysis uses all open-weight pairs
with identifiable bad-state marginals, including same-family pairs
(\sFourNPairs{} pairs). Within each item-bootstrap replicate, marginals and
pair membership are re-estimated before the pair summaries are calculated.

\textbf{Primary-protocol vote sets.} Figure~2 estimates correlations between
each model's first temperature-0.7 sample. The protocol contrasts in
Fig.~\ref{fig:smreason} and Table~\ref{tab:samepair} use greedy votes under
both protocols. Each comparison therefore uses one sampling protocol on
both sides.

\subsection{Adequacy of the one-factor approximation}
\label{sm:factor}
For each domain and state we fit a rank-one (single-factor) structure to
the identifiable block of $\hat\Sigma_\theta$ by iterated loadings on the
off-diagonals (\texttt{codes/e7\_bootstrap.py}) and report the mean
absolute off-diagonal residual. Residuals span
\factorResidMin--\factorResidMax{} across the eight domain $\times$ state
blocks (mean \factorResidTyp) against mean absolute off-diagonal
magnitudes of \factorOffdiagMin--\factorOffdiagMax{} (mean
\factorOffdiagTyp). A single common factor absorbs most of the shared
structure, supporting the exchangeable summaries quoted in the main text.
Committee predictions use the full estimated matrix $\hat\Sigma_\theta$ to
retain the remaining pair-specific dependence (\S\ref{sm:est}).

\subsection{Reasoning-protocol re-collection}
\label{sm:reason}
To measure protocol sensitivity, every item was re-voted under a \emph{reasoning
protocol}: the prompt presents the same evidence and invites step-by-step
verification, the response budget is raised to 1{,}536--2{,}048 tokens
for API judges (per collection script) and 1{,}024 tokens for local
judges, provider thinking channels are enabled where available,
and the verdict is parsed from the final \texttt{VERDICT: approve|reject}
line of the answer. The fallback searches the last approve/reject keyword in
the answer and then the reasoning trace. Any remaining unparsed vote is coded
as reject under the primary-protocol rule; the overall remaining-unparsed rate is
\unparsedThinkPct\%. The judges are all \nModelsFrontier{}
frontier models (full item sets per domain) and eleven open-weight models
spanning six families and 0.5B--34B (Qwen2.5 \{0.5B, 1.5B, 7B, 14B, 32B\},
Gemma-2 \{2B, 9B\}, Mistral-7B, Llama-3.1-8B, GLM-4-9B, Yi-34B).
Collection is greedy (temperature 0), \nVotesReasonExact{} votes in total.

Estimation follows the primary protocol: greedy votes, per-pair common
items, tetrachoric inversion with the same identifiability guards
($\ge 8$ items per state, marginals in $(0.05, 0.95)$), and a paired
stratified item bootstrap ($B = 1{,}000$; the same resampled items are
applied to both protocols, so protocol deltas are within-replicate).
The protocol contrast compares means over each protocol's identifiable
pairs. Within every item-bootstrap replicate, marginals and pair sets are
recomputed separately for the two protocols. Table~\ref{tab:samepair}
additionally restricts both estimates to pairs identifiable under both
protocols.

\subsection{Additional reasoning-protocol checks}
\label{sm:robust}
\textbf{Interaction test and power.} On the same-pair-intersection
estimator of Table~\ref{tab:samepair}, a paired bootstrap ($B = 1{,}000$;
one bad-state item resample per replicate applied to all four
roster $\times$ protocol cells) gives the difference-in-differences
$\Delta_{\text{frontier}} - \Delta_{\text{open}}$ directly, and the
bootstrap SE yields the minimal detectable effect of the frontier delta at
80\% power (two-sided 5\%; $\mathrm{MDE} = 2.80 \times \mathrm{SE}$).
\textbf{Self-resampling under reasoning.} The five-model intervention
committee re-votes the first \selfCotItems{} facts items (103 $G$ / 97 $B$)
under the reasoning protocol, once greedy and in three temperature-0.7
samples (\selfCotVotes{} votes; \texttt{outputs/decisions\_selfreason/}).
Self-correlation per model and state is the mean tetrachoric over the three
sample pairs with the standard guards, estimated identically on the primary
corpus's three terse-protocol temperature-0.7 samples of the same items. The
two protocols are compared model for model and item for item.

\subsection{Compute and reproducibility}
The primary corpus used in the main analyses contains
\nVotesBaseExact{} votes over three sessions. The local open-weight models were
served on one NVIDIA A100 80GB GPU; the eight API models were queried through
the provider and gateway routes described above. The local portion of the
\nVotesReasonExact{}-vote reasoning-protocol re-collection used the same
A100 80GB GPU, alongside the frontier APIs.
Table~\ref{tab:reproducibility} links each reported result to its script and
output.

\begingroup
\centering\scriptsize\setlength{\tabcolsep}{3pt}
\begin{longtable}{@{}>{\raggedright\arraybackslash}p{0.30\textwidth} >{\raggedright\arraybackslash}p{0.32\textwidth} >{\raggedright\arraybackslash}p{0.28\textwidth}@{}}
\caption{\textbf{Reproducibility map.} Empirical estimands are linked to
the scripts and archived output files from which the reported values are
obtained.}
\label{tab:reproducibility}\\
\hline
Reported result & Script & Output \\
\hline
\endfirsthead
\hline
Reported result & Script & Output \\
\hline
\endhead
\hline
\multicolumn{3}{r}{\textit{Continued on next page}}\\
\endfoot
\hline
\endlastfoot
Theory numerical checks (Table~\ref{tab:cert}) & \texttt{e1\_theory\_cert.py} & \texttt{e1\_theory\_cert.json} \\
Vote collection & \texttt{run\_decisions.py} & \texttt{decisions/*.jsonl} \\
$\hat\Sigma$, calibration, selection (Figs.~2, 3) & \texttt{analyze\_votes.py} & \texttt{analysis\_*\_full.json} \\
API vote collection & \texttt{run\_decisions\_api.py} & \texttt{decisions/*.jsonl} \\
Bootstrap intervals for correlation summaries and floors & \texttt{e7\_bootstrap.py} & \texttt{e7\_bootstrap.json} \\
Reasoning votes & \texttt{run\_reason\_recheck.py} & \texttt{decisions/*\_\_think.jsonl} \\
Reasoning co-error (Fig.~S5) & \texttt{e10\_reason\_coerror.py}, \texttt{e11\_reason\_open.py} & \texttt{e10/e11 *.json} \\
Family vs.\ capability & \texttt{e12\_s4\_family.py} & \texttt{e12\_s4\_family.json} \\
Same-pair robustness (Table~\ref{tab:samepair}) & \texttt{e13\_samepair.py} & \texttt{e13\_samepair.json} \\
Interaction test / power & \texttt{e15\_did\_power.py} & \texttt{e15\_did\_power.json} \\
Self-resampling under reasoning & \texttt{run\_selfreason.py}, \texttt{e16\_selfreason.py} & \texttt{decisions\_selfreason/*}, \texttt{e16\_selfreason.json} \\
Sampled-composition prediction calibration and omissions (Fig.~3A--C, Table~\ref{tab:ablation}) & \texttt{e19\_existing\_robustness.py} & \path{e19_existing_robustness.json} \\
Existing-vote size curves (Table~\ref{tab:sizecurves}) & \texttt{e20\_existing\_size\_curves.py} & \path{e20_existing_size_curves.json} \\
Reviewed-record parser sensitivity & \path{e22_parser_audit.py}; \path{e24_parser_adjudication.py}; \path{reviewed_record_uncertainty.py} & \path{reviewed_record_sensitivity_index.json} \\
Design-balanced selection inference (Table~\ref{tab:selection}) & \path{e23_full_pipeline_bootstrap.py} & \path{design_balanced_selection_index.json} \\
Selection decomposition and selector agreement (Table~\ref{tab:selection}) & \path{selection_summary.py} & \path{selection_summary.json} \\
\end{longtable}
\endgroup

\clearpage
\section{Extended results}
\label{sm:results}

\subsection{Per-domain correlation structure}
Figure~\ref{fig:smheat} shows state-conditional latent correlation matrices
for all four domains (blank cells: pairs outside the tetrachoric
identifiability window). Some high-correlation pairs, especially in facts
and truthfulness, have greater bad-state dependence than good-state
dependence (main text Fig.~2C), a pattern compatible with shared failure
modes. The two verification domains where many small models
sit near chance (math, code) illustrate this identifiability limit. Nearly
constant voting patterns yield unstable pairwise latent correlations.
Marginal behavior and the other members' joint behavior remain relevant to
threshold design.

\begin{figure}[h]
\centering
\includegraphics[width=\textwidth]{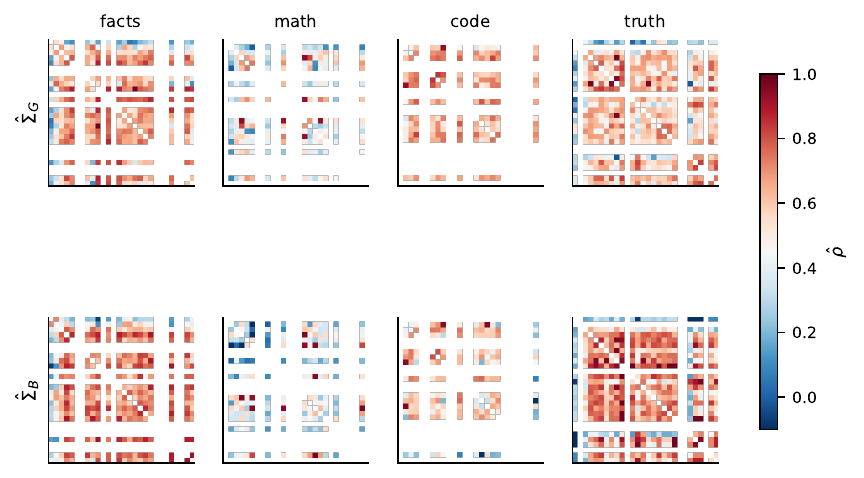}
\caption{\textbf{State-conditional latent correlation matrices.}
Latent error-correlation matrices $\hat\Sigma_G$ (top row) and
$\hat\Sigma_B$ (bottom row) across domains, estimated between the
first temperature-0.7 vote of each model on all items of the domain.
These descriptive matrices use all domain items; Fig.~3 uses the odd/even
split for validation. Blank cells mark pairs outside the tetrachoric
identifiability window.}
\label{fig:smheat}
\end{figure}

\subsection{Per-domain calibration}
\begin{figure}[h]
\centering
\includegraphics[width=\textwidth]{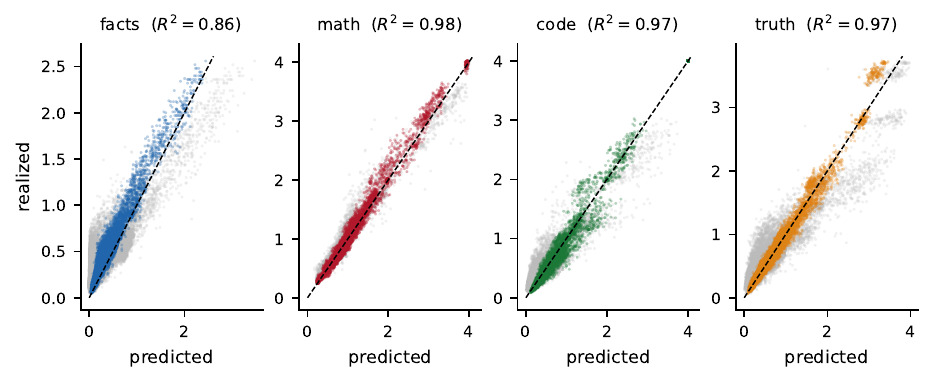}
\caption{\textbf{Per-domain calibration on held-out items.}
Predicted versus realized \kofn{} loss by domain. Colored points show
heterogeneous full-matrix predictions, gray points show independence
predictions, and the dashed line is the identity line. The panels contain
73,410 threshold--cost rows clustered by committee because they share items
and fitted parameters.}
\label{fig:smcal}
\end{figure}

\subsection{Secondary reasoning-protocol check}
Figure~\ref{fig:smreason} compares bad-state co-error under the primary and
reasoning protocols (methods in \S\ref{sm:reason}). For the eight frontier
model labels, reasoning changes balanced accuracy across domains
(\rsnFrBalaccMin--\rsnFrBalaccMax). The 95\% CIs for changes in
$\hat\rho_B$ include zero in the identifiable facts and truthfulness
comparisons: truthfulness \rsnFrTruthBaseB{} $\to$
\rsnFrTruthB{} ($\Delta = \rsnFrTruthDelta$, 95\% CI
\rsnFrTruthDeltaCI) and fact verification \rsnFrFactsBaseB{} $\to$
\rsnFrFactsB{} ($\Delta = \rsnFrFactsDelta$, CI
\rsnFrFactsDeltaCI). The available identifiable pairs support panel-level
contrasts in facts and truthfulness. For the eleven-model open-weight pool, the
point estimates decrease in all four domains: facts
($\Delta = \rsnOpFactsDelta$, CI \rsnOpFactsDeltaCI), code
($\Delta = \rsnOpCodeDelta$, CI \rsnOpCodeDeltaCI), truthfulness
($\Delta = \rsnOpTruthDelta$, CI \rsnOpTruthDeltaCI), and math
($\Delta = \rsnOpMathDelta$, CI \rsnOpMathDeltaCI).

These exploratory estimates characterize protocol-associated changes within
the observed panels. The paired difference-in-differences estimates are
$\didFactsDelta$ (95\% CI \didFactsDeltaCI) in fact verification and
\didTruthDelta{} (CI \didTruthDeltaCI) in truthfulness. Both intervals
include zero, and the frontier delta's 80\%-power minimal detectable effect
is \mdeFrFacts{} in facts and \mdeFrTruth{} in truthfulness.
Table~\ref{tab:samepair} reports estimates on the pairs identifiable under
both protocols.

\begin{figure}[H]
\centering
\includegraphics[width=\textwidth]{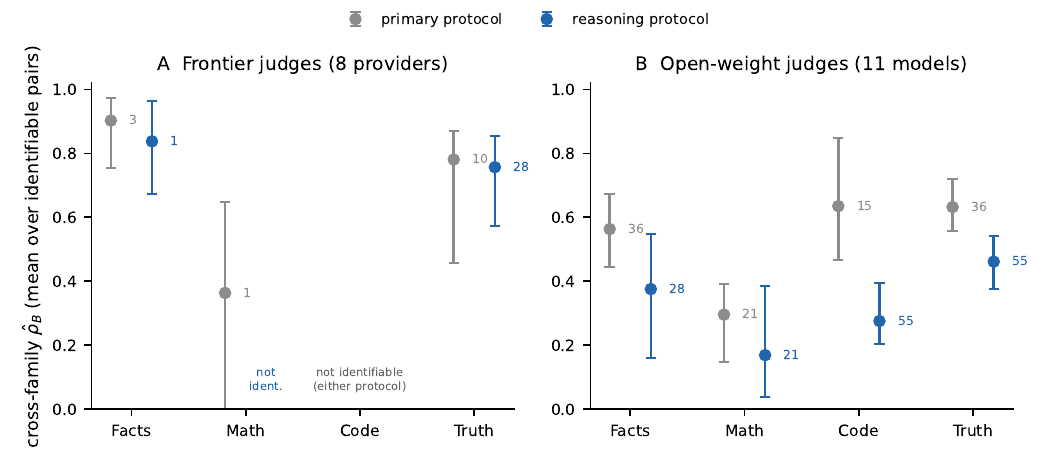}
\caption{\textbf{Exploratory reasoning-protocol contrasts.} Bad-state
latent correlation $\hat\rho_B$ under the primary (gray) and reasoning
(blue) protocols, estimated from greedy votes (\S\ref{sm:methods}). Points
average over the identifiable pairs for each protocol, and error bars are
95\% paired stratified-item bootstrap intervals ($B=1{,}000$); small integers
give the corresponding pair counts. Panel A shows the API-model panel. Facts
and truthfulness yield contrasts with wide intervals; code yields zero
identifiable pairs under both protocols, and mathematics yields one pair
under the primary protocol. Panel B shows the eleven-model open-weight pool.
Pair sets differ across panels; the estimated within-panel changes and their
intervals provide the comparable contrasts.}
\label{fig:smreason}
\end{figure}

The sampled-reasoning check compares three temperature-0.7 samples for a
five-model facts committee with the corresponding terse-protocol samples
(methods in \S\ref{sm:robust}). The pooled self-correlation estimate changes
from the terse protocol's clip boundary (\selfCotBTerse) to
$\hat\rho_B = \selfCotB$ (95\% CI \selfCotBCI; identifiable for
\selfCotIdentB{} of \selfCotModels{} judges; $\Delta = \selfCotDB$, CI
\selfCotDBCI). The resulting point estimate remains above the cross-model
reasoning-pool estimate on the same domain ($\hat\rho_B = \rsnOpFactsB$).
This estimate pertains to the five evaluated models, the facts domain, and
the sampled reasoning protocol; one component lies on the clip boundary.

\begin{table}[H]
\centering\small
\begin{tabular}{llccc}
\hline
Pool & Domain & pairs (primary/reasoning/both) & $\Delta\hat\rho_B$ on shared pairs & 95\% CI \\
\hline
Frontier & Facts & 3 / 1 / 1 & \spFrFactsDelta & \spFrFactsDeltaCI \\
Frontier & Truthfulness & 10 / 28 / 10 & \spFrTruthDelta & \spFrTruthDeltaCI \\
Open (11) & Facts & 36 / 28 / 28 & \spOpFactsDelta & \spOpFactsDeltaCI \\
Open (11) & Math & 21 / 21 / 15 & \spOpMathDelta & \spOpMathDeltaCI \\
Open (11) & Code & 15 / 55 / 15 & \spOpCodeDelta & \spOpCodeDeltaCI \\
Open (11) & Truthfulness & 36 / 55 / 36 & \spOpTruthDelta & \spOpTruthDeltaCI \\
\hline
\end{tabular}
\caption{\textbf{Same-pair robustness of the protocol contrast.} Mean
$\hat\rho_B(\text{reasoning}) - \hat\rho_B(\text{primary})$ over the pairs
identifiable under \emph{both} protocols, with paired stratified item
bootstrap 95\% CIs ($B = 1{,}000$). Marginals and identifiability
are re-evaluated within each replicate inside the restricted candidate
universe. Frontier code yields zero identifiable pairs under both protocols;
frontier math yields one primary-protocol pair.}
\label{tab:samepair}
\end{table}

\subsection{Family and capability summary}
This descriptive analysis defines each model's capability as sample-1
balanced accuracy over all complete items in a domain,
$\tfrac12[\hat p_G+(1-\hat p_B)]$. Fig.~2B's raw comparison gives mean
bad-state correlation
\sFourRawSame{} for same-family pairs and \sFourRawCross{} for cross-family
pairs over \sFourNPairs{} identifiable pairs. The groups differ in model
scale, capability, and domain composition, making the raw contrast jointly
associated with these characteristics. In the specified regression, which
controls for pairwise capability gap, mean capability, and domain fixed
effects, the same-family coefficient is \sFourCoef{} (95\%
model-node-bootstrap CI \sFourCoefCI). Restricting to pairs with
$|\Delta\text{cap}| \le 0.10$ gives a gap of \sFourMatchedGap{} (CI
\sFourMatchedGapCI). Under these specifications, both uncertainty intervals
include zero; same- and cross-family complementarity are statistically
indistinguishable for this roster.

\begin{table}[H]
\centering\small\setlength{\tabcolsep}{5pt}
\begin{tabular}{lccc}
\hline
Prediction model & Identity-line $R^2$ (95\% CI) & RMSE & Mean error \\
\hline
Heterogeneous full-matrix dependence & \robustFullRtwo{} \robustFullRtwoCI & \robustFullRMSE & \robustFullBias \\
Nonnegative exchangeable one-factor projection & \robustExchangeRtwo{} \robustExchangeRtwoCI & \robustExchangeRMSE & \robustExchangeBias \\
Independence & \robustIndepRtwo{} \robustIndepRtwoCI & \robustIndepRMSE & \robustIndepBias \\
\hline
\end{tabular}

\vspace{0.8em}
\begin{tabular}{lccc}
\hline
Omission analysis & Full-matrix $R^2$ range & Independence $R^2$ range & Gain over majority range \\
\hline
Leave one domain out & \robustLodoFullRange & \robustLodoIndepRange & \robustLodoGainRange \\
Leave one model family out & \robustLofoFullRange & \robustLofoIndepRange & \robustLofoGainRange \\
Leave one model out & \robustLomoFullRange & \robustLomoIndepRange & \robustLomoGainRange \\
\hline
\end{tabular}
\caption{\textbf{Sampled-composition prediction ablation and omission
sensitivity.} The upper panel uses all \robustAblationN{} held-out-item
threshold--cost rows. Whole committee draws form the resampling clusters
($B=\robustBootB$). The lower panel
gives ranges across each omission set under the same sampled composition.
Gain is majority loss minus the loss of the heterogeneous-model-selected
threshold in native loss units. These checks quantify sensitivity to
excluding observed domains, model families, or individual models.}
\label{tab:ablation}
\end{table}

\begin{table}[H]
\centering\footnotesize\setlength{\tabcolsep}{3pt}
\resizebox{\textwidth}{!}{%
\begin{tabular}{lcccc}
\hline
Domain & $L_3$ (95\% CI) & $L_{25}$ (95\% CI) & $L_{25}-L_3$ (95\% CI) & Nested rosters with $L_{25}<L_3$ \\
\hline
Fact verification & \sizeFactsThreeLoss{} \sizeFactsThreeCI & \sizeFactsTwentyFiveLoss{} \sizeFactsTwentyFiveCI & \sizeFactsDeltaTwentyFiveMinusThree{} \sizeFactsDeltaCI & \sizeFactsTwentyFiveBetterPct\% \\
Mathematics & \sizeMathThreeLoss{} \sizeMathThreeCI & \sizeMathTwentyFiveLoss{} \sizeMathTwentyFiveCI & \sizeMathDeltaTwentyFiveMinusThree{} \sizeMathDeltaCI & \sizeMathTwentyFiveBetterPct\% \\
Code review & \sizeCodeThreeLoss{} \sizeCodeThreeCI & \sizeCodeTwentyFiveLoss{} \sizeCodeTwentyFiveCI & \sizeCodeDeltaTwentyFiveMinusThree{} \sizeCodeDeltaCI & \sizeCodeTwentyFiveBetterPct\% \\
Truthfulness & \sizeTruthThreeLoss{} \sizeTruthThreeCI & \sizeTruthTwentyFiveLoss{} \sizeTruthTwentyFiveCI & \sizeTruthDeltaTwentyFiveMinusThree{} \sizeTruthDeltaCI & \sizeTruthTwentyFiveBetterPct\% \\
\hline
\end{tabular}}
\caption{\textbf{Finite committee-size curves from existing greedy votes.}
Each of \sizeCurveRosters{} random nested rosters is evaluated at odd sizes
$n=3,5,\ldots,25$. Intervals are stratified item-bootstrap intervals
($B=\sizeCurveBootB$); roster-to-roster dispersion is archived in
\texttt{e20\_existing\_size\_curves.json}. Endpoint contrasts and improvement
fractions compare $n=25$ directly with $n=3$ within the same nested roster.}
\label{tab:sizecurves}
\end{table}

\begin{table}[H]
\centering\footnotesize\setlength{\tabcolsep}{4pt}
\begin{tabular}{@{}p{0.43\textwidth}cc@{}}
\hline
Quantity & Point estimate & 95\% percentile CI \\
\hline
Dependence-aware benefit over majority (scaled loss units, $\times100$)
  & \SelectionPrimaryAbsoluteBenefitVsMajorityScaledUnits{}
  & [\SelectionPrimaryAbsoluteBenefitVsMajorityCILowScaledUnits{},
     \SelectionPrimaryAbsoluteBenefitVsMajorityCIHighScaledUnits{}] \\
Dependence-aware benefit over majority (\%)
  & \SelectionPrimaryRelativeBenefitPct{}
  & [\SelectionPrimaryRelativeBenefitCILowPct{},
     \SelectionPrimaryRelativeBenefitCIHighPct{}] \\
Dependence-aware benefit over the independence-based threshold
  (scaled loss units, $\times100$)
  & \SelectionPrimaryAbsoluteBenefitVsIndependenceScaledUnits{}
  & [\SelectionPrimaryAbsoluteBenefitVsIndependenceCILowScaledUnits{},
     \SelectionPrimaryAbsoluteBenefitVsIndependenceCIHighScaledUnits{}] \\
Independence-based benefit over majority
  (scaled loss units, $\times100$)
  & \SelectionIndependenceBenefitVsMajorityScaledUnits{}
  & [\SelectionIndependenceBenefitVsMajorityCILowScaledUnits{},
     \SelectionIndependenceBenefitVsMajorityCIHighScaledUnits{}] \\
Independence-based benefit over majority (\%)
  & \SelectionIndependenceRelativeBenefitPct{}
  & [\SelectionIndependenceRelativeBenefitCILowPct{},
     \SelectionIndependenceRelativeBenefitCIHighPct{}] \\
Dependence-aware/independence-based threshold agreement (\%)
  & \SelectionFullIndependenceThresholdAgreementPct{}
  & [\SelectionFullIndependenceThresholdAgreementCILowPct{},
     \SelectionFullIndependenceThresholdAgreementCIHighPct{}] \\
\hline
\end{tabular}

\vspace{0.8em}
\begin{tabular}{@{}p{0.68\textwidth}c@{}}
\hline
Additional point estimate & Value \\
\hline
Dependence-aware selected scaled loss ($\times100$)
  & \SelectionPrimaryFullScaledLoss{} \\
Independence-based selected scaled loss ($\times100$)
  & \SelectionPrimaryIndependenceScaledLoss{} \\
Majority scaled loss ($\times100$)
  & \SelectionPrimaryMajorityScaledLoss{} \\
Share of the dependence-aware benefit supplied by threshold choice under
  independence (\%) & \SelectionThresholdOnlySharePct{} \\
Share supplied by dependence modeling (\%)
  & \SelectionDependenceSharePct{} \\
Sampled-composition relative benefit over majority (\%)
  & \SelectionArchivedCompositionRelativeBenefitPct{} \\
Pool-kind-balanced relative benefit over majority (\%)
  & \SelectionPoolKindBalancedRelativeBenefitPct{} \\
Self-only relative benefit over majority (\%)
  & \SelectionSelfOnlyRelativeBenefitPct{} \\
Cross/mixed relative benefit over majority (\%)
  & \SelectionCrossMixedRelativeBenefitPct{} \\
Indeterminate reviewed records assigned reject (\%)
  & \SelectionForcedRejectPointRelativeBenefitPct{} \\
Agreement-only reviewed-record overlay (\%)
  & \SelectionAgreementOnlyPointRelativeBenefitPct{} \\
\hline
\end{tabular}
\caption{\textbf{Design-balanced threshold-selection results.} The primary
estimand gives equal weight to 20
domain--pool-kind--size cells and fully refits and reselects thresholds in
each of \SelectionBootstrapValid{} item $\times$ committee-cluster replicates using
\SelectionMonteCarloDraws{} antithetic common Monte Carlo draws. The upper
panel reports bootstrap intervals from the refitted and reselected
replicates. The lower panel gives the corresponding loss decomposition,
weighting summaries, and parser-policy summaries defined in
\S\ref{sm:parseraudit}.}
\label{tab:selection}
\end{table}

\paragraph{Regularization diagnostics.}
Across the design-balanced point fit,
\SelectionFrechetSaturationPct{}\% of eligible
tetrachoric inversions saturated a Fr\'echet boundary and
\SelectionUnbracketedClipPct{}\% used an unbracketed clip.
\SelectionPSDProjectionPct{}\% of committee state matrices required
positive-semidefinite projection; the mean and maximum Frobenius projection
distances were \SelectionPSDProjectionMeanDistance{} and
\SelectionPSDProjectionMaxDistance{}. Every full-pipeline replicate repeats
the tetrachoric inversions, matrix projection, loss prediction, and threshold
selection, carrying this finite-sample variation into the reported loss and
threshold intervals.

\subsection{CAPA versus latent correlation}
\begin{figure}[h]
\centering
\includegraphics[width=0.6\textwidth]{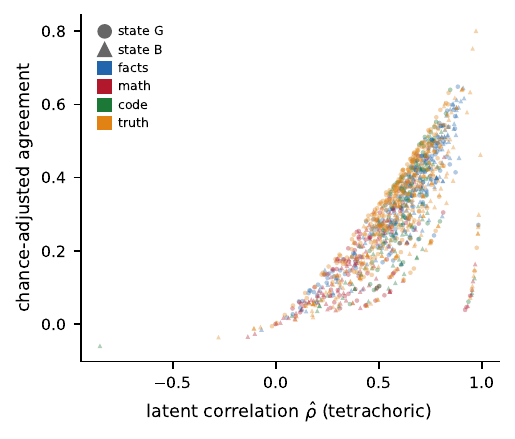}
\caption{\textbf{Chance-adjusted agreement versus latent correlation.}
Chance-adjusted agreement (binary CAPA analogue) against
tetrachoric $\hat\rho$ for 1,186 identifiable pair--domain--state
observations. At fixed state-specific marginals CAPA is monotone in $\rho$.
Changing marginals changes this transform and places empirical pairs on
different latent-correlation scales. Threshold design uses the
state-conditional latent quantity.}
\label{fig:smcapa}
\end{figure}

\clearpage
\section{Extended related work}
\label{sm:related}

\paragraph{Organizational economics of fallible decision-making.}
The hierarchy and polyarchy screening framework originates with Sah and
Stiglitz \cite{sah1985human,sah1986architecture}, who compared architectures
by their type-I/type-II error composition under independent errors, and was
extended to committees and thresholds in \cite{sah1988committees}. Team
theory \cite{marschak1972} founded the broader information-in-organizations
program. Modern treatments compute reliability-optimal decision structures
\cite{christensen2010} and organizational exploration frontiers
\cite{csaszar2013}. State-conditional correlation enters the screening
formulation as a design input alongside state-specific error rates and
$\kappa$.

\paragraph{Jury theorems under dependence.}
Condorcet's independent-vote argument motivates majority aggregation
\cite{condorcet1785}. Under dependence, Ladha bounds majority reliability
\cite{ladha1992}; Boland, Proschan and Tong analyze majority systems with a
common leader-follower dependence \cite{boland1989}; Berend and Sapir study
monotonicity \cite{berend2007}; Kaniovski and Zaigraev treat optimal jury
design with exchangeable correlated votes \cite{kaniovski2011}; Nitzan and
Paroush give the optimal weighted rule for independent heterogeneous jurors
\cite{nitzan1982}. Together these literatures establish how dependence
changes jury reliability and how heterogeneity changes optimal rules under
independence. The LLM screening problem combines member-specific marginals
with state-specific dependence; their fitted joint vote distributions
determine loss-based selection among \kofn{} thresholds.

\paragraph{Ensemble diversity in machine learning.}
That ensemble gains require disagreement is classical. The
error--ambiguity decomposition \cite{krogh1995}, diversity--accuracy studies
\cite{kuncheva2003,dietterich2000}, and Breiman's random-forest bound
$\mathrm{PE} \le \bar\rho\,(1-s^2)/s^2$ \cite{breiman2001} all quantify
versions of it, as does the wisdom-of-crowds literature
\cite{hong2004,davisstober2014}. Cost-sensitive screening adds a decision
dimension to diversity measurement: false acceptance and false rejection
enter loss separately, and state-conditional dependence determines the loss
of each candidate threshold for a fixed roster.

\paragraph{LLM multi-agent systems.}
Self-consistency \cite{wang2023selfconsistency} and sample-and-vote scaling
\cite{li2024moreagents} aggregate a single model; debate
\cite{du2024debate,liang2024divergent,chan2024chateval}, mixture-of-agents
\cite{wang2024moa}, and role-structured frameworks
\cite{wu2023autogen,hong2024metagpt} aggregate several. Evaluation stacks
use single LLM judges \cite{zheng2023judge} or juries \cite{verga2024jury}.
Kim et al.\ \cite{kim2025scaling} report a controlled study and an explicitly
predictive model (260 configurations, six
benchmarks; cross-validated $R^2 = 0.373$, or $0.413$ with a task-grounded
capability metric), finding coordination costs and non-monotone returns to
agent count. We examine the complementary question of using state-conditional
error correlation as a committee-design input. Strong dependence provides one
mechanism for non-monotone returns because additional agents may add little
information.

\paragraph{Correlated errors and AI oversight.}
Goel et al.\ \cite{goel2025great} introduce chance-adjusted probabilistic
agreement (CAPA) and report associations between model capability,
similarity, and judge bias; Kim et al.\ \cite{kim2025correlated} document
correlated errors across a large model panel, including provider and
capability patterns, and study a hiring application. These measured
similarity patterns supply state-conditional inputs to committee loss. State
conditioning is consequential because good-state misses and bad-state false
acceptances enter the screening objective separately.

\paragraph{Monoculture and systemic homogeneity.}
Kleinberg and Raghavan model welfare losses when many actors adopt the same
algorithm \cite{kleinberg2021}; Bommasani et al.\ measure outcome
homogenization from shared foundations \cite{bommasani2022}; Hammond et
al.\ catalogue correlated-failure risks in multi-agent AI
\cite{hammond2025multiagent}; model-collapse dynamics couple models through
shared training data \cite{shumailov2024}. In the screening model, this
dependence enters the floor \eqref{eq:floorratio} and the collapse theorem.

\paragraph{Mathematical tools.}
The one-factor Gaussian threshold model and its large-pool limit are due to
Vasicek \cite{vasicek2002}, with regulatory foundations by Gordy
\cite{gordy2003}; Slepian's inequality \cite{slepian1962}, Plackett's
identity \cite{plackett1954}, tetrachoric correlation \cite{pearson1900},
convex-order machinery \cite{shaked2007}, and Topkis's monotone comparative
statics \cite{topkis1998} supply the proof infrastructure. Gunn et al.\
\cite{gunn2016} analyzed how unanimous verdicts can become implausible in
forensic systems; our
Remark~\ref{rem:toogood} derives its organizational analogue endogenously
from unequal state correlations.

\FloatBarrier
\bibliographystyle{unsrtnat}
\bibliography{arxiv_manuscript}

\end{document}